\documentclass[10pt, journal, onecolumn]{IEEEtran}
\usepackage{textcomp}
\usepackage[colorlinks,linkcolor=red,anchorcolor=blue,citecolor=blue]{hyperref}
\usepackage{tikz}
\usepackage{amsthm}
\usepackage{amsmath}
\usepackage[linesnumbered,ruled,vlined]{algorithm2e}
\usepackage{graphicx}
\usepackage{url}
\usetikzlibrary{arrows,positioning, calc}
\tikzstyle{vertex}=[draw,fill=black!15,circle,minimum size=18pt,inner sep=0pt]
\usepackage{amsmath,amssymb}
\usepackage{multicol}
\usepackage{multirow}
\usepackage{pgfplots}
\usepackage[utf8]{inputenc}
\usepackage{authblk}
\usepackage[FIGBOTCAP]{subfigure}
\usetikzlibrary{shapes.geometric}
\usepackage{adjustbox}
\usepackage{bm}
\usepackage[utf8]{inputenc}
\usepackage{amsthm}
\usetikzlibrary{decorations.pathreplacing,calligraphy}

\newtheorem{theorem}{Theorem}
\newtheorem*{remark}{Remark}
\newtheorem{lemma}{Lemma}

\usetikzlibrary{shapes.callouts} 
\definecolor{darkgreen}{rgb}{0.0, 0.5, 0.0}
\pgfkeys{%
    /calloutquote/.cd,
    width/.code                   =  {\def\calloutquotewidth{#1}},
    position/.code                =  {\def\calloutquotepos{#1}}, 
    author/.code                  =  {\def\calloutquoteauthor{#1}},
    /calloutquote/.unknown/.code   =  {\let\searchname=\pgfkeyscurrentname
                                 \pgfkeysalso{\searchname/.try=#1,                                
    /tikz/\searchname/.retry=#1},\pgfkeysalso{\searchname/.try=#1,
                                  /pgf/\searchname/.retry=#1}}
                            }  

\newcommand{\asymcloud}[2][.1]{%
\begin{scope}[#2]
\pgftransformscale{#1}%
\pgfpathmoveto{\pgfpoint{261 pt}{115 pt}} 
  \pgfpathcurveto{\pgfqpoint{70 pt}{107 pt}}
                 {\pgfqpoint{137 pt}{291 pt}}
                 {\pgfqpoint{260 pt}{273 pt}} 
  \pgfpathcurveto{\pgfqpoint{78 pt}{382 pt}}
                 {\pgfqpoint{381 pt}{445 pt}}
                 {\pgfqpoint{412 pt}{410 pt}}
  \pgfpathcurveto{\pgfqpoint{577 pt}{587 pt}}
                 {\pgfqpoint{698 pt}{488 pt}}
                 {\pgfqpoint{685 pt}{366 pt}}
  \pgfpathcurveto{\pgfqpoint{840 pt}{192 pt}}
                 {\pgfqpoint{610 pt}{157 pt}}
                 {\pgfqpoint{610 pt}{157 pt}}
  \pgfpathcurveto{\pgfqpoint{531 pt}{39 pt}}
                 {\pgfqpoint{298 pt}{51 pt}}
                 {\pgfqpoint{261 pt}{115 pt}}
\pgfusepath{fill,stroke}         
\end{scope}}    

\newcommand\calloutquote[2][]{%
       \pgfkeys{/calloutquote/.cd,
         width               = 5cm,
         position            = {(0,-1)},
         author              = {}}
  \pgfqkeys{/calloutquote}{#1}                   
  \node [rectangle callout,callout relative pointer={\calloutquotepos},text width=\calloutquotewidth,/calloutquote/.cd,
     #1] (tmpcall) at (-1,2) {#2};
  \node at (tmpcall.pointer){\calloutquoteauthor};    
}  
\makeatletter
\def\url@leostyle{%
  \@ifundefined{selectfont}{\def\UrlFont{\sf}}{\def\UrlFont{\small\ttfamily}}}
\makeatother
\newcommand{\mygrid}{\tikz{ \draw [ultra thick, draw=black, step=1cm] (0,0) grid  (5,3) rectangle (0,0);
}
}
\newcommand{\mygridd}{\tikz{ \draw [ thick, draw=black, step=1cm] (0,0) grid  (8,0.5) rectangle (0,0);
}
}

\newcommand{\mytwogrids}{\tikz{ \draw [ thick, draw=black, step=4cm] (0,0) grid  (18,1) rectangle (0,0);
}
}
\newcommand{\myfourgrids}{\tikz{ \draw [ thick, draw=black, step=1cm] (0,0) grid  (4,0.5) rectangle (0,0);
}
}

\begin{document}

\bstctlcite{IEEEexample:BSTcontrol}

\title{In-Network Volumetric DDoS Victim Identification \\Using Programmable Commodity Switches }
\author{Damu Ding\href{https://orcid.org/0000-0001-9692-7756}{\protect\includegraphics[scale=0.1]{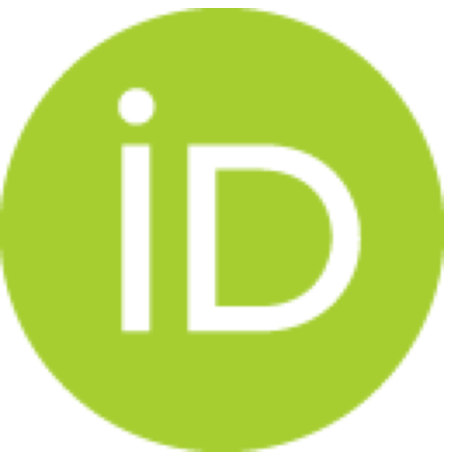}}, \IEEEmembership{Student Member, IEEE},
        Marco Savi\href{https://orcid.org/0000-0002-8193-0597}{\protect\includegraphics[scale=0.1]{orcid/icon.eps}},
        Federico Pederzolli\href{https://orcid.org/0000-0003-0557-9878}{\protect\includegraphics[scale=0.1]{orcid/icon.eps}},\\
        \vspace{-12pt}
        Mauro Campanella\href{https://orcid.org/0000-0002-0228-8701}{\protect\includegraphics[scale=0.1]{orcid/icon.eps}}, and
        Domenico Siracusa\href{https://orcid.org/0000-0002-5640-6507}{\protect\includegraphics[scale=0.1]{orcid/icon.eps}}
\thanks{Damu Ding is with Fondazione Bruno Kessler, Trento, Italy and University of Bologna, Bologna, Italy.
}
\thanks{Marco Savi is with University of Milano-Bicocca, Milano, Italy. The study was mainly done while he was with Fondazione Bruno Kessler, Trento, Italy.}
\thanks{Federico Pederzolli and Domenico Siracusa are with Fondazione Bruno Kessler, Trento, Italy.}
\thanks{Mauro Campanella is with GARR, Roma, Italy.} 
}

\maketitle

\makeatletter
\def\ps@IEEEtitlepagestyle{%
  \def\@oddfoot{\mycopyrightnotice}%
  \def\@oddhead{\hbox{}\@IEEEheaderstyle\leftmark\hfil\thepage}\relax
  \def\@evenhead{\@IEEEheaderstyle\thepage\hfil\leftmark\hbox{}}\relax
  \def\@evenfoot{}%
}
\def\mycopyrightnotice{%
  \begin{minipage}{\textwidth}
  \centering \scriptsize
  Copyright~\copyright~2021 IEEE. Personal use of this material is permitted. Permission from IEEE must be obtained for all other uses, in any current or future media, including\\reprinting/republishing this material for advertising or promotional purposes, creating new collective works, for resale or redistribution to servers or lists, or reuse of any copyrighted component of this work in other works by sending a request to pubs-permissions@ieee.org.
  \end{minipage}
}
\makeatother
\begin{abstract}
Volumetric distributed Denial-of-Service (DDoS) attacks have become one of the most significant threats to modern telecommunication networks.
However, most existing defense systems require that detection software operates from a centralized monitoring collector, leading to increased traffic load and delayed response. 
The recent advent of Data Plane Programmability (DPP) enables an alternative solution: threshold-based volumetric DDoS detection can be performed directly in programmable switches to skim only potentially hazardous traffic, to be analyzed in depth at the controller. 
In this paper, we first introduce the BACON data structure based on sketches, to estimate per-destination flow cardinality, and theoretically analyze it. 
Then we employ it in a simple in-network DDoS victim identification strategy, INDDoS, to detect the destination IPs for which the number of incoming connections exceeds a pre-defined threshold. 
We describe its hardware implementation on a Tofino-based programmable switch using the domain-specific P4 language, proving that some limitations imposed by real hardware to safeguard processing speed can be overcome to implement relatively complex packet manipulations.
Finally, we present some experimental performance measurements, showing that our programmable switch is able to keep processing packets at line-rate while performing volumetric DDoS detection, and also achieves a high F1 score on DDoS victim identification.

\end{abstract}
\begin{IEEEkeywords}
 Anomaly detection, Programmable data planes, DDoS victim identification, P4
\end{IEEEkeywords}
\section{Introduction}

Distributed Denial-of-Service (DDoS) attacks are a critical security threat to modern telecommunication networks; not only do they cripple live services for legitimate users, but also cause large operational burdens on operators, which must dedicate significant resources to detecting and mitigating them. 
As the number and size of botnets and DDoS attacks persistently increases \cite{santannajjIM2015}, so does this workload. In particular, \emph{volumetric} DDoS attacks, designed to overwhelm network and server capacity, are among the most common and dangerous DDoS attacks \cite{ramanathan2018senss}.

Many techniques to both perform and detect volumetric DDoS attacks are documented in the scientific literature.
Among the latter, a rather common feature of such attacks is the exploitation of a (large) number of (capacious) source hosts to direct a considerable amount of packets to a specific victim destination \cite{yu2013software}\cite{liu2016one}\cite{huang2017sketchvisor}. 
Attackers use seemingly legitimate TCP, UDP, or ICMP packets in volumes large enough to overwhelm network devices and servers, or deliberately incomplete packets designed to rapidly consume all available computing, storage, and transmission resources in servers. 
A majority of such attacks also make use of spoofed source IP addresses (e.g. DNS and NTP amplification DDoS attacks), that is, they forge the address that supposedly generated the requests to prevent source identification; this implies that tracking attack source IPs is usually ineffective.
On the other hand, DDoS victim addresses cannot be spoofed, and identifying the victims is a useful step for network operators to mitigate such attacks.

In legacy networks, including SNMP- \cite{van2014opennetmon} and NetFlow-based \cite{Netflow} networks, flow monitoring and anomaly/attack detection are executed on top of a centralized \emph{monitoring collector} or \emph{controller}. 
Such a collector needs to periodically analyze the sampled flow statistics retrieved from network devices to monitor the network state, which has two well-known drawbacks \cite{ding2020incrementally}: \emph{(i.)} significant communication overhead is generated between switches and the centralized monitor/controller due to frequent flow sampling, and \emph{(ii.)} large detection latency, which limits the system efficiency in detecting relatively short attacks lasting one or two minutes.
Furthermore, it limits the monitoring system ability to detect features that manifest even in sparsely sampled traffic (which may or may not be a problem in the specific case of volumetric DDoS attack detection, but is nonetheless a limitation).
The advent of \emph{data-plane programmable switches} (e.g., Tofino \cite{Tofino}) enables the migration of volumetric DDoS attack detection (and victim identification) to the switches themselves, using domain-specific programming languages like P4 \cite{bosshart2014p4}. 
Despite using a relatively high-level language, programming the behaviour of a physical switch requires specific care in the optimization of physically limited resources to maintain line-rate processing speeds.
The language specifications reflect the hardware boundaries: it features the absence of expensive and unbounded operations (e.g., division and loops), due to the limited number of pipeline stages in the processor \cite{qian2019flexgate}.
Such constraints place a ceiling on the number of operations executable on each packet to ensure line-speed processing.

This paper describes the implementation and validation of volumetric DDoS victim detection and identification directly in Tofino-based \cite{Tofino} P4-enabled commodity switches.
The proposed logic copes well with the aforementioned restrictions of data plane programmability.
That is, unlike other recent works on this subject, we actually implemented and tested our technique in a real switch instead of limiting the work to a P4 simulator, such as the Behavioral model \cite{p4simulator}.
To that end, after describing some background mathematical results and data structures required by our technique in Section~\ref{sec:background}, we propose, in Section~\ref{sec:inddos}, two main contributions: \emph{BACON Sketch} and \emph{INDDoS}. 
BACON is a new sketch (a probabilistic data structure) combining a Direct Bitmap \cite{estan2003bitmap} and a Count-min Sketch \cite{cormode2005improved}, which allows switches to estimate the number of distinct flows (i.e., packets with the same flow key) contacting the same destination host. 
INDDoS is a simple volumetric DDoS victim identification strategy built on top of BACON Sketch to identify the destination IPs contacted by a number of source IPs  greater than a  threshold, in a given time interval, completely in the data plane.
We include extensive theoretical analysis and detail the modifications required to implement our approach in physical resource-constrained P4 Tofino switches in Section \ref{subsec:impl}.
Then, Section~\ref{sec:integration} provides insights on the integration of INDDoS in a full DDoS defense system, including attack detection and consequent mitigation steps. 
Section~\ref{sec:evaluation} presents the results of an evaluation of INDDoS in a small testbed.
We show that, using optimal parameters derived from our theoretical analysis, our implementation can reach an F1 score higher than 0.95 on a real flow trace \cite{caida} captured on a 10\,Gbps backbone link, without performance degradation in the switch's packet-processing capabilities.
Finally, Section~\ref{sec:related} relates our paper to other recent works on the subject while Section~\ref{sec:conclusion} concludes it. 
The work has been carried out within the G\'EANT~\cite{GEANT} GN4-3 project.

\section{Background} \label{sec:background}
\subsection{Markov's inequality}
 Given a non-negative random variable $X$ and a positive value $a$, Markov's inequality defines a constant upper bound for the probability that satisfies $\mathbb{P}(X \ge a) \le \frac{\mathbb{E}[X]}{a}$, where $\mathbb{E}[X]$ represents the expected value of $X$. 

\subsection{Direct Bitmap}
Direct Bitmap \cite{estan2003bitmap} is a simple data structure that can be used to estimate the number of distinct flows occurring in a packet stream (also called \emph{flow cardinality}): it is based on a bit array called \emph{Bitmap register} and relies on one or more \emph{hash functions}. 
Initially, all $m$ cells in the Bitmap register are set to 0. 
When a packet arrives, its flow key is hashed: the hashed key indicates the index of the register cell to consider.
The value of the indexed cell is set to 1 if it was previously 0, otherwise it is not updated. 
Note that packets sharing the same flow key are always hashed into the same cell, while different flows are hashed to different cells unless a collision occurs.
The number of distinct flows can then be (under-)estimated by counting the number of bits with value 1 in the register.

\subsection{Count-min Sketch}
The estimation of the \emph{per-flow packet count} (i.e., number of packets carried by any network flow during an observation interval) can be performed using a \emph{Count-min Sketch} \cite{cormode2005improved}, a probabilistic and memory-efficient data structure which implements \emph{Update} and \emph{Query} operations: the former is responsible for continuously updating the sketch to count incoming packets in the switch, while the latter retrieves the estimated number of incoming packets for a specific flow. 
Count-min Sketch relies on $d$ pairwise-independent hash functions, each with an output of size $w$. 
The data structure is composed by a matrix of $d \cdot w$ counters:
the accuracy of packet count estimation in Count-min Sketch increases as $d$ or $w$ increase, and vice versa.

\section{In-network DDoS victim identification} \label{sec:inddos}
\subsection{Threat model and deployment scenario}

\subsubsection{Threat model}
In this paper, we focus on volumetric DDoS attacks against victim destinations in the network. 
For the purpose of overwhelming the available resources of the victims, an attacker exploits a large number of distributed hosts (e.g. bots in a botnet) to frequently send traffic to the target host(s) (e.g., a web server).
The attacker sources are usually spoofed to evade detection, which, coupled with the fact that each source may send only a small amount of traffic to the victim (in a stealthy volumetric DDoS attack), makes identifying attack sources (also called \emph{superspreaders}) less efficient than focusing on destinations for DDoS mitigation.
Unfortunately, there exists legitimate network events, such as flash crowds, that exhibit similar characteristics to volumetric DDoS attacks, making a suspiciously high number of sources contacting a destination a necessary but not sufficient condition to determine whether a destination is under attack.
To further discriminate DDoS attacks from flash crowds, a possible solution can be found in \cite{yu2011discriminating}, where the correlation of flows is used to determine whether they are malicious or not.

\input{scenario}
\subsubsection{Deployment scenario for the proposed DDoS detection}
We target an ISP network, for which the best placement of our DDoS detection functionality involves deploying programmable switches at the edge of the network, so that at least one switch has visibility on all flows towards each IP destination (\figurename~\ref{fig:scenario}).
Therefore, at least one switch is in the best place to estimate the number of source hosts contacting any destination host.
Once a DDoS victim is identified, ideally after a more thorough analysis step at a centralized controller, any border programmable switch can also be used to limit the traffic rate towards it. 

\subsection{BACON Sketch}
For the purpose of measuring flow cardinality directly inside switches, we combined Direct Bitmap registers \cite{estan2003bitmap} and Count-min Sketch \cite{cormode2005improved} in a new sketch, which we named BACON (BitmAp COuNt-min) Sketch for the sake of brevity. 
As shown in \figurename~\ref{fig:bacon}, in BACON Sketch the counters of Count-min Sketch are replaced with a $m$-sized Bitmap register, hence the size of BACON is $d \times w \times m$, where $d \times w$ is the size of Count-min Sketch. 
BACON Sketch enables the estimation of per-destination flow cardinality, with flows identified by different flow keys $key_{src}$, using very little memory in the switch. 
Two different flow keys are considered for each packet: the  key $key_{src}$ must include the source IP of the flow together with any subset of \{source port, destination IP and port, protocol\} without loss of generality, and the choice of $key_{src}$ depends on the requirements of operators. 
Likewise, the flow key of the destination host, denoted by $key_{dst}$, can be either the destination IP or the \{destination IP, destination port\} pair. 

Formally, BACON Sketch solves the following problem.
Given: 
\begin{itemize}
    \item a packet stream $S$
    \item a Bitmap register size $m$
    \item a number of hash functions in Count-min Sketch $d$
    \item an output size of hash functions in Count-min Sketch $w$
    \item a time interval $T_{int}$
    \item a flow key $key_{dst}$
\end{itemize}
compute $\hat{E}_{dst}$, the estimated value of $E_{dst}$ in $T_{int}$, where $E_{dst}$ is the number of flows (identified by different $key_{src}$ keys) that contact the destination host identified by $key_{dst}$.
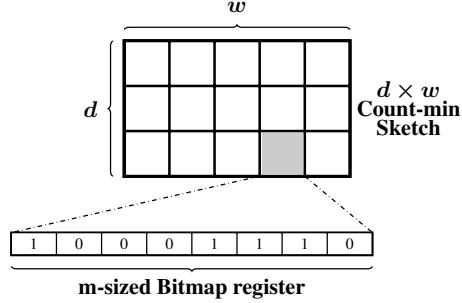
\begin{figure}[t]
\centering
\scalebox{0.6}{
\begin{tikzpicture}[]
\node(CMS) at (4.5, 2.5)[rectangle, draw, inner sep=0]{
\mygrid
};
\node (kk) at (8.3,2.5) [minimum width=2cm,minimum height=1cm, align=center] {\Large{$\bm{d\times w}$} \\\Large{\textbf{Count-min}} \\ \Large{\textbf{Sketch}}};
\node [fill=black!20!white, minimum size=0.9cm] at (5.5, 1.5) {};
\node(HLLm) at (3.5, -0.5)[rectangle, draw, inner sep=0]{
\mygridd
};
\draw[thick,decorate,decoration={brace,amplitude=3pt,mirror, raise = 1cm}]
            (-0.5,0) -- (7.5,0) ;
\node(register) at (3.5,-1.5){ \Large{\textbf{m-sized Bitmap register}}};
\node(value) at (7,-0.5){0};
\node(value) at (6,-0.5){1};
\node(value) at (5,-0.5){1};
\node(value) at (4,-0.5){1};
\node(value) at (3,-0.5){0};
\node(value) at (2,-0.5){0};
\node(value) at (1,-0.5){0};
\node(value) at (0,-0.5){1};

\draw[-, dash dot, thick] (5, 1) -- (-0.5, -0.25);
\draw[-, dash dot,thick] (6, 1) -- (7.5, -0.25);

\draw[thick,decorate,decoration={brace,amplitude=3pt}]
            (1.75, 1) -- (1.75,4) ;
\node(d) at (1.25,2.5){\Large{$\bm{d}$}};
    
\draw[thick,decorate,decoration={brace,amplitude=3pt}]
            (2, 4.25) -- (7,4.25) ;            
\node(w) at (4.5,4.75){\Large{$\bm{w}$}};
\end{tikzpicture}
}
\caption{Data structure of BACON Sketch}
\label{fig:bacon}
\end{figure}

\begin{algorithm}[t]
\caption{\textbf{BACON Sketch}}
\label{alg:bacon}
\SetKwProg{Fn}{Function}{:}{}
\KwIn{Packet stream $S$, where each packet is characterized by $key_{src}$ and $key_{dst}$}
\KwOut{Estimated number of distinct flows contacting the packet's destination $key_{dst}$ (i.e., $\hat{E}_{dst}$)} 
$d \leftarrow$ Number of hash functions $h_{cm}^{i}$ in Count-min\\
$w \leftarrow$ Output size of hash functions $h_{cm}^{i}$ in Count-min\\
$m \leftarrow$ Bitmap-register size with hash function $h_{bm}$\\
\For{Each packet in $S$ (with $key_{src}$ and $key_{dst}$)}{
    \textit{Update}($key_{src}$, $key_{dst}$)\\
    \textit{Query}($key_{dst}$)\\
}

\Fn {Update($key_{src}$, $key_{dst}$)}{
	$bucket\leftarrow h_{bm}(key_{src}) \% m$\label{line:bucket}\\
	\For{Each hash function $h_{cm}^{i}$\label{line:upd1}}{ 
	    $index  \leftarrow (h_{cm}^{i}(key_{dst})\%w) \cdot m + bucket$\\
		\If{$BACON_{i}[index] == 0$}{
		$BACON_{i}[index] \leftarrow 1$ \label{line:upd2} \tcp{row $i$}
		}
		}
}

\Fn{Query($key_{dst}$)}{
    $\hat{E}_{dst} \leftarrow 0$	\\
    \For{Each hash function $h_{cm}^{i}$\label{line:qr1}}{
            $id \leftarrow (h_{cm}^{i}(key_{dst})\%w) \cdot m$\\
            $E_{i} \leftarrow \sum_{j=id }^{id + m - 1}{BACON_{i}[j]} $\\
            		\If{$\hat{E}_{dst} == 0$}{
			$ \hat{E}_{dst} \leftarrow$ $E_{i}$\\
		}\ElseIf{$E_{i}$ $<$ $\hat{E}_{dst} $}{
			 $ \hat{E}_{dst} \leftarrow$ $E_{i}$\label{line:qr2}
		}
}
		\Return{$\hat{E}_{dst}$}
}

\end{algorithm}

Algorithm \ref{alg:bacon} shows the pseudo code of BACON Sketch.
For each incoming packet, the switch hashes the flow key $key_{src}$ with hash function $h_{bm}$ and converts the hashed value to be within range $[0, m-1]$: this value, named $bucket$, is the index in the Bitmap register (Line~\ref{line:bucket}).
Then each of the $d$ pairwise-independent hash function in BACON, denoted by $h_{cm}^{i}$, hashes the packet's $key_{dst}$ to the slot $(h_{cm}^{i}$($key_{dst})\% w)\cdot m + bucket$ in each row, setting the related register cell to 1 (Lines~\ref{line:upd1}-\ref{line:upd2}). $BACON_{i}$ refers to the $i$-th row of BACON Sketch.
Note that this way to compute the cell's index to be updated assumes that the Bitmap registers' cells in each row are progressively numbered from 0 to $m\cdot w -1$, much like in an array.
Concerning the \emph{Query} operation for any $key_{dst}$, the involved Bitmap register in row $i$ (i.e., the register indexes from $(h_{cm}^{i}(key_{dst})\%w)\cdot m $ to $h_{cm}^{i}(key_{dst})\%w)\cdot m + m - 1$) estimates the cardinality of $key_{dst}$ by computing the sum of values in the register cells (i.e., number of 1s). This is done for each row and the minimum estimated cardinality $\hat{E}_{dst}$ among all $d$ rows (i.e. $\hat{E}_{dst} = min(E_{i})$) is returned (Lines~\ref{line:qr1}-\ref{line:qr2}).

\subsection{In-network cardinality-based DDoS victim identification}
Using BACON Sketch, we propose a simple in-network volumetric DDoS detection mechanism, which we named INDDoS, aiming to identify likely DDoS victims, i.e., those hosts that are contacted by an abnormally high number of source IPs, and hence the associated network flows.

Formally, the problem is formulated as follows.\\
Given: 
\begin{itemize}
    \item a number of source IPs $n$
    \item a BACON Sketch with size $d \times w \times m$ $(m < n)$
    \item a DDoS threshold fraction $\theta$ (threshold: $\theta n$)
    \item a time interval $T_{int}$
\end{itemize}
return all destinations $key_{dst}$, named \emph{DDoS victims}, that satisfy $\hat{E}_{dst} > \theta n$ within time interval $T_{int}$, where $\hat{E}_{dst}$ is the estimated number of sources contacting $key_{dst}$ and obtained by querying the BACON Sketch. 

From this point onward, without any loss of generality, we will use $src$ (the source IP) as $key_{src}$, and $dst$ (the destination IP) as $key_{dst}$, to estimate how many flows from different source IPs are trying to contact a destination IP in a given time interval $T_{int}$.
Observe that the proposed strategy, which requires updating and querying the BACON Sketch for each incoming packet towards any destination \emph{dst} during the time interval, interacts with a remote controller only when at least one attack is detected in the same interval, otherwise no communication between controller and switches takes place. 

\subsection{Theoretical analysis}
In this section we present a theoretical analysis of the accuracy of our proposed DDoS detection strategy.

\subsubsection{Error bounds for estimated flow cardinality}

\begin{theorem}\label{th:1}
With  probability at least $1 - ({\frac{1}{2})}^{d}$, the cardinality estimation $\hat{E}_{dst}$ from BACON Sketch satisfies $ E_{dst} -2(n - m( 1 - e^{-\frac{n}{m}})) < \hat{E}_{dst}$.
\end{theorem}
\begin{proof}
Let us start with the Bitmap register. 
Since a register cell is touched (i.e., set to 1) as soon as it is selected for one incoming element, the probability of any cell $X_{i}$ to be touched after processing one element (packet) can be expressed as:
\begin{align*}
     \mathbb{P}(X_{i} = 1) &= \frac{1}{m}
\end{align*}
so, the probability of $i$-th register cell to be untouched is:
\begin{align*}
     \mathbb{P}(X_{i} = 0) &= 1 - \mathbb{P}(X_{i} = 1) = 1 -\frac{1}{m}.
\end{align*}
After processing $E_{dst}$ different elements (where $E_{dst}$ is the true number of $src$ contacting $dst$) we then have:
\begin{align*}
\mathbb{P}(X_{i} = 0) &= (1 - \frac{1}{m})^{E_{dst}} = ((1 - \frac{1}{m})^{m})^{\frac{E_{dst}}{m}}.
\end{align*}
Since $\lim_{m\rightarrow  \infty}{(1-\frac{1}{m})^{m}} = e^{-1}$ and in our case $m$ is large (i.e., $\ge$ 128), ${(1-\frac{1}{m})^{m}} \approx e^{-1}$ and we have: 
\begin{align*}
\mathbb{P}(X_{i} = 0)\approx  e^{-\frac{E_{dst}}{m}}.
\end{align*}
Hence, the probability of a cell to be touched after $E_{dst}$ elements is: 
\begin{align*}
\mathbb{P}(X_{i} = 1) &=  1 - \mathbb{P}(X_{i} = 0) = 1- e^{-\frac{E_{dst}}{m}}.
\end{align*}
Since there are $m$ cells in each Bitmap register, the expectation of our Bitmap cardinality-estimation (i.e., the expected number of 1s in the register), for destination IP $dst$,  is:
\begin{align*}
\mathbb{E}[\hat{E}_{dst}^{Bitmap}] &= \sum_{i=1}^{m}{X_{i} \cdot \mathbb{P}(X_{i} = 1) + X_{i} \cdot \mathbb{P}(X_{i} = 0)}\\
&= \sum_{i=1}^{m}{(1 \cdot \mathbb{P}(X_{i} = 1)) + 0}\\
&= \sum_{i=1}^{m}{1- e^{-\frac{E_{dst}}{m}}} =  m(1- e^{-\frac{E_{dst}}{m}}).
\end{align*}
Note that, when there are no collisions in Count-min Sketch, $\mathbb{E}[\hat{E}_{dst}] = \mathbb{E}[\hat{E}_{dst}^{Bitmap}]$.
Since $\hat{E}_{dst}^{Bitmap} \le E_{dst}$, the difference between $\mathbb{E}[\hat{E}_{dst}]$ and $\mathbb{E}[E_{dst}]$\footnote{The expectation of $E_{dst}$ is still $E_{dst}$ since $E_{dst}$ is a constant.} is:
\begin{align*}
    f(E_{dst}) =\mathbb{E}[E_{dst}] - \mathbb{E}[\hat{E}_{dst}] &= E_{dst} - m(1- e^{-\frac{E_{dst}}{m}}).
\end{align*}
The derivative of $f(E_{dst})$, denoted by $f'(E_{dst})$, is:
\begin{align*}
    f'(E_{dst}) = 1 - e^{-\frac{E_{dst}}{m}}.
\end{align*}
Since $e^{-\frac{E_{dst}}{m}} < 1$, $f'(E_{dst})>0$. 
For this reason, $f(E_{dst})$ is monotonically increasing and $f(E_{dst}) \le f(n) = n - m( 1 - e^{-\frac{n}{m}})$.
Then, applying Markov's inequality yields:
\begin{align*}
        \mathbb{P}(E_{dst} - \hat{E}_{dst} &\ge 2(n - m( 1 - e^{-\frac{n}{m}})))\\ &\le \frac{E_{dst} -\mathbb{E}[\hat{E}_{dst}]}{2(n - m( 1 - e^{-\frac{n}{m}}))}\\
        &\le \frac{n - m( 1 - e^{-\frac{n}{m}})}{2(n - m( 1 - e^{\frac{-n}{m}}))}
        = \frac{1}{2}.
\end{align*}
Considering that there are $d$ hash functions in the Count-min Sketch part of BACON Sketch:
\begin{align*}
\mathbb{P}((E_{dst} - \hat{E}_{dst}) \ge 2(n - m( 1 - e^{-\frac{n}{m}}))) &\le  (\frac{1}{2})^{d}
\end{align*}
Thus, with probability at least $1 - ({\frac{1}{2})}^{d}$,  
\begin{align*}\label{eq1}
   \hat{E}_{dst} >  E_{dst} -2(n - m( 1 - e^{-\frac{n}{m}})). && \qedhere
\end{align*}
\end{proof}

\begin{theorem}\label{th:2}
With  probability at least $1 - ({\frac{1}{2})}^{d}$, the cardinality estimation $\hat{E}_{dst}$ from BACON Sketch satisfies $\hat{E}_{dst} \le E_{dst} + 2m(1-e^{-\frac{n}{mw}})$.
\end{theorem}

\begin{proof}
Count-min Sketch occasionally writes two different inputs to the same cell; the expected number of such collisions is $\frac{n - E_{dst}}{w}$ \cite{cormode2005improved}.
Therefore, in BACON Sketch, the actual average number of different elements being written to the same Bitmap register is $E_{dst} + \frac{n - E_{dst}}{w} < E_{dst} + \frac{n}{w}$, so the expectation of $\hat{E}_{dst}$ becomes:
\begin{align*}
\mathbb{E}[\hat{E}_{dst}] = \sum_{i=1}^{m}{\mathbb{E}[X_{i}]} &= \sum_{i=1}^{m}{1- e^{-\frac{E_{dst} + \frac{n - E_{dst}}{w}}{m}}}\\
&=  m(1- e^{-\frac{E_{dst}+ \frac{n - E_{dst}}{w}}{m}})\\
& < m(1- e^{-\frac{E_{dst}+ \frac{n}{w}}{m}}).
\end{align*}
Then, the expectation of $\hat{E}_{dst} - \hat{E}^{Bitmap}_{dst}$ is:
\begin{align*}
      \mathbb{E}[\hat{E}_{dst}-\hat{E}^{Bitmap}_{dst}]               &= \mathbb{E}[\hat{E}_{dst}] - \mathbb{E}[\hat{E}^{Bitmap}_{dst}] \\
                    &< m(1- e^{-\frac{E_{dst}+ \frac{n}{w}}{m}}) - m(1- e^{-\frac{E_{dst}}{m}})\\
                    &= m e^{-\frac{E_{dst}}{m}}(1-e^{-\frac{n}{mw}})
                    < m(1-e^{-\frac{n}{mw}}).
\end{align*}
According to Markov's inequality, 
\begin{align*}
        \mathbb{P}((\hat{E}_{dst} - \hat{E}^{Bitmap}_{dst}) &\ge 2m(1-e^{-\frac{n}{mw}}))\\ &\le \frac{\mathbb{E}[\hat{E}_{dst}] -\mathbb{E}[\hat{E}_{dst}^{Bitmap}]}{2m(1-e^{-\frac{n}{mw}})}\\
        &= \frac{m(1-e^{-\frac{n}{mw}})}{2m(1-e^{-\frac{n}{mw}})}
        = \frac{1}{2}.
\end{align*}
Since there are $d$ hash functions in Count-min Sketch:
\begin{align*}
        \mathbb{P}((\hat{E}_{dst} - \hat{E}_{dst}^{Bitmap}) \ge 2m(1-e^{-\frac{n}{mw}})) \le (\frac{1}{2})^{d}.
\end{align*}
Hence, with probability at least $1 - ({\frac{1}{2})}^{d}$, 
\begin{align*}
       \hat{E}_{dst} - \hat{E}_{dst}^{Bitmap} < 2m(1-e^{-\frac{n}{mw}}).
\end{align*}
Since $\hat{E}_{dst}^{Bitmap} \le E_{dst}$, $\hat{E}_{dst}$ satisfies:
\begin{align*}
    \hat{E}_{dst}  < 2m(1-e^{-\frac{n}{mw}}) + E_{dst}.  && \qedhere
\end{align*}
\end{proof}
\begin{remark}
The lower error bound for the cardinality estimation is $E_{dst} -2(n - m( 1 - e^{-\frac{n}{m}}))$, while the upper error bound is $E_{dst} + 2m(1-e^{-\frac{n}{mw}})$. 
This implies that we should carefully choose $m$ and $w$ instead of setting the sketch width (i.e., number of columns) as large as possible, due to their counteracting effect on estimation quality. 
\end{remark}

\subsubsection{Error bounds for DDoS detection}
\begin{lemma}\label{lemma:1}
With  probability at least $1 - ({\frac{1}{2})}^{d}$, $R \le \frac{2n}{w}$, where $R$ is the overestimation of $E_{dst}$ caused by Count-min Sketch.
\end{lemma}
\begin{proof}
The expectation of $R$ in each Bitmap register is $\mathbb{E}[R] = \frac{n-E_{dst}}{w} \le \frac{n}{w}$, applying Markov's inequality yields:
\begin{align*}
    \mathbb{P}(R \ge \frac{2n}{w}) \le (\frac{1}{2})^{d}.  && \qedhere
\end{align*}
\end{proof}

Given Lemma~\ref{lemma:1}, Theorem~\ref{th:3} reports the \emph{false negative bound} and Theorem~\ref{th:4} the \emph{false positive bound} for INDDoS.

\begin{theorem}\label{th:3}
When $E_{dst} \ge \theta n + 2(n - m(1-e^{-\frac{n}{m}}))$, INDDoS reports dst as a DDoS victim with a probability of at least $1 - ({\frac{1}{2})}^{d}$.
\end{theorem}

\begin{proof}
A victim should be reported if $E_{dst}>\theta n$.
By Theorem~\ref{th:1}, $\hat{E}_{dst} \ge E_{dst} - 2(n - m(1-e^{-\frac{n}{m}})) \ge \theta n$ with probability at least $1 - ({\frac{1}{2})}^{d}$. Therefore, $dst$ is reported as a DDoS victim with the same probability if and only if $\hat{E}_{dst} \ge E_{dst} -  2(n - m(1-e^{-\frac{n}{m}}))$.
\end{proof}

\begin{theorem}\label{th:4}
When $E_{dst} \le \frac{2}{w}n$ and $\theta \ge \frac{4}{w}$, BACON Sketch reports $dst$ as a victim with a probability of at most $ ({\frac{1}{2})}^{d}$.
\end{theorem}
\begin{proof}
\begin{align*}
    \mathbb{P}(\hat{E}_{dst} \ge \theta n) 
    &\le \mathbb{P}(E_{dst} + R \ge \theta n) \\
    &\le \mathbb{P}(R + \frac{2n}{w} \ge \theta n) \qquad (\textit{due to } E_{dst} \le \frac{2}{w}n)\\
    &= \mathbb{P}(R \ge (\theta - \frac{2}{w})n) \\
    &\le \mathbb{P}(R \ge (\frac{4}{w} - \frac{2}{w})n) \qquad (\textit{due to } \theta \ge \frac{4}{w})\\
   &=\mathbb{P}(R \ge \frac{2}{w}n) \le ({\frac{1}{2})}^{d} \qquad (\textit{by } Lemma~\ref{lemma:1}).  \qedhere
\end{align*}
\end{proof}
\begin{remark}
The false negative and positive bounds for INDDoS show that increasing the BACON Sketch height (i.e., number of rows) improves the victim identification performance (with equal column parameters).
\end{remark}

\section{Implementation in commodity switches} \label{subsec:impl}
\begin {figure}[t]
\centering
\scalebox{0.45}{
\begin{tikzpicture}
\centering

\node(stream) at (-1.5,2.5)[draw, align=center]  {\huge{\textbf{Packet}} \\ \huge{\textbf{stream $S$}}};
\node (HLL) at (3.5,2.5) [draw,thick,minimum width=2cm,minimum height=1cm, align=center] { \Large{$\bm{bucket = }$}\\ \Large{$\bm{h_{bm}(key_{src})\% m}$}  };

\node(CMS) at (13.5, 2.5)[rectangle, draw, inner sep=0]{
\mygrid
};
\node (kk) at (14.75,4.5) [minimum width=2cm,minimum height=1cm] {\Large{\textbf{BACON Sketch}}};

\draw [->](stream) -- (HLL) node [midway, above, sloped] (TextNode) {\Large{$\bm{key_{src}}$}};
\draw [-](6.8, 2.5) -- (6.8, 3.5) ;
\node (block1)[fill=black!20!white, minimum size=0.9cm] at (12.5, 3.5) {};

\draw [->](6.8, 3.5) -- (12.5, 3.5) node [left=3.6cm, above] (TextNode) {\large{$\bm{h_{cm}^{1}(key_{dst}), bucket}$}};
\node [fill=black!20!white, minimum size=0.9cm] at (11.5, 2.5) {};
\draw [->](HLL) -- (11.5, 2.5) node [left=2.6cm, above, sloped] (TextNode) {\large{$\bm{h_{cm}^{2}(key_{dst}), bucket}$}};
\draw [->](HLL) -- (11.5, 2.5) node [midway, below, sloped] (TextNode) {\large{$\bm{\cdots}$}};

\node [fill=black!20!white, minimum size=0.9cm] at (14.5, 1.5) {};

\draw [-](6.8, 2.5) -- (6.8, 1.5) ;
\draw [->](6.8, 1.5) -- (14.5, 1.5) node [left=5.6cm, above, sloped] (TextNode) {\large{$ \bm{ h_{cm}^{d}(key_{dst}), bucket}$}};

\node(HLLm) at (12.5, -0.5)[rectangle, draw, inner sep=0]{
\mygridd
};
\draw[thick,decorate,decoration={brace,amplitude=3pt,mirror, raise = 1cm}]
            (8.5,0) -- (16.5,0) ;
\node(register) at (12.5,-1.5){\Large{\textbf{An m-sized Bitmap register}}};
\draw[->,thick] (13, 0.25) -- (13,-0.25);
\node(bucket) at (13, 0.37){\Large{$\bm{bucket}$}};
\node(value) at (13,-0.5){1};
\draw[-, dash dot, thick] (14, 1) -- (8.5, -0.25);
\draw[-, dash dot,thick] (15, 1) -- (16.5, -0.25);

\node(count) at (13.5, 7.5)[rectangle, draw, inner sep=0]{
\mygrid
};
\node (counterN) at (14.75,9.5) [minimum width=2cm,minimum height=1cm] {\Large{\textbf{Count-min Sketch}}};
\node (b1)[fill=black!20!white, minimum size=0.9cm] at (12.5, 8.5) {\textbf{+0/1}};
\node (b2)[fill=black!20!white, minimum size=0.9cm] at (11.5, 7.5) {\textbf{+0/1}};
\node (b3)[fill=black!20!white, minimum size=0.9cm] at (14.5, 6.5) {\textbf{+0/1}};

\node(es1)[draw, align=center] at (8, 8.5){\Large{\textbf{Estimated  cardinality}} \large{$\bm{E_{1}}$}};
\node(es2)[draw, align=center] at (8, 7.5){\Large{\textbf{Estimated  cardinality} $\bm{E_{2}}$}};
\node(es3)[draw, align=center] at (8, 6.5){\Large{\textbf{Estimated  cardinality} $\bm{E_{d}}$}};
\node[align=center] at (8, 7){$\bm{\cdots}$};

\draw[->](b1)--(es1);
\draw[->](b2)--(es2);
\draw[->](b3)--(es3);

\draw [->](16.5, -0.5) edge[bend right=40, draw=blue] node [left] {} (b3);
\draw [->](12.5, 3.5)edge[bend right, draw=blue] node [left] {} (b1);
\draw [->](11.5, 2.5)edge[bend left=10, draw=blue] node [left] {} (b2);
\draw[thick,decorate,decoration={brace,amplitude=3pt, raise = 0.7cm}]
            (5.8,6) -- (5.8,9) ;
\node(min) at (4.5,7.5){$\Large{\bm{\hat{E}_{dst}}}$};
\node(checkTd)[draw, align=center] at (2.75, 7.5){\Large{$\bm{\hat{E}_{dst} =}$} \\ \Large{$\bm{\theta n + 1?}$}}; 
\draw[->] (min) -- (checkTd);

\node(controller)[draw,minimum height=1cm] at (-1.5, 7.5){\huge \textbf {Controller}}; 
\draw[->] (checkTd) -- (controller)node [midway, align=center] (TextNode) {\Large{\textbf{Report}}\\\Large{$\bm{key_{dst}}$}};
\end{tikzpicture}
}
\caption{Scheme of  INDDoS including Update (bottom) and Query (top) operations on BACON Sketch as implemented in the commodity switch}
\label{DDoS}
\end{figure}
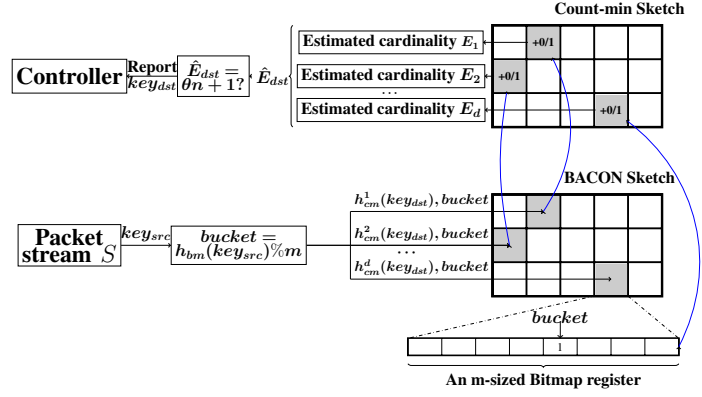
\figurename~\ref{DDoS} shows a schematic representation of our in-network DDoS detection implementation within P4-enabled commodity switches. 
We implemented INDDoS with 680 lines of P4\_16 code, released as open source at \cite{P4INDDoS}.
Here we report the details of the implementation of BACON Sketch and INDDoS in the Tofino-based hardware switch.
\subsection{Tofino-based switch data plane architecture}
\begin{figure}[t]
\centering
\begin{adjustbox}{width=0.6\linewidth}
\begin{tikzpicture}[]
\node(TM) at (0, 0)[rectangle, draw, inner sep=1, line width=1mm,  minimum width = 2.5cm, minimum height=5cm, align=center]{\textbf{\Large{Traffic}} \\ \textbf{\Large{manager}}};

\node(in1) at (-2.5, 2)[rectangle, draw, inner sep=1, line width=1mm,  minimum width = 2.5cm, minimum height=1cm,fill=black!20!white]{Pipe 1};
\node(in2) at (-2.5, 0)[rectangle, draw, inner sep=1, line width=1mm,  minimum width = 2.5cm, minimum height=1cm,fill=black!20!white]{Pipe $k$};
\node(in3) at (-2.5, -2)[rectangle, draw, inner sep=1, line width=1mm,  minimum width = 2.5cm, minimum height=1cm,fill=black!20!white]{Pipe $M$};
\node(in4) at (-2.5, -3)[align=center]{\textbf{Ingress}\\\textbf{Pipeline}};
\node(in4) at (-4.3, -3)[align=center]{\textbf{Ingress}\\\textbf{port}};

\draw[-, dashed](-1.25, -2) -- (-1.25, -3.5);
\draw[-, dashed](-3.75, -2) -- (-3.75, -3.5);

\node(ic1)  at(-2.5, 1) {\bm{$\cdots$}};
\node(ic2)  at(-2.5, -1) {\bm{$\cdots$}};

\node(p1) at (-4, 2.4)[rectangle, draw, inner sep=1, line width=1mm,  minimum width = 0.5cm, minimum height=0.2cm,fill=black]{};
\node(p2) at (-4, 2)[rectangle, draw, inner sep=1, line width=1mm,  minimum width = 0.5cm, minimum height=0.2cm,fill=black]{};
\node(p3) at (-4, 1.6)[rectangle, draw, inner sep=1, line width=1mm,  minimum width = 0.5cm, minimum height=0.2cm,fill=black]{};

\node(p4) at (-4, 0.4)[rectangle, draw, inner sep=1, line width=1mm,  minimum width = 0.5cm, minimum height=0.2cm,fill=black]{};
\node(p5) at (-4, 0)[rectangle, draw, inner sep=1, line width=1mm,  minimum width = 0.5cm, minimum height=0.2cm,fill=black]{};
\node(p6) at (-4, -0.4)[rectangle, draw, inner sep=1, line width=1mm,  minimum width = 0.5cm, minimum height=0.2cm,fill=black]{};

\node(p7) at (-4, -1.6)[rectangle, draw, inner sep=1, line width=1mm,  minimum width = 0.5cm, minimum height=0.2cm,fill=black]{};
\node(p8) at (-4, -2)[rectangle, draw, inner sep=1, line width=1mm,  minimum width = 0.5cm, minimum height=0.2cm,fill=black]{};
\node(p9) at (-4, -2.4)[rectangle, draw, inner sep=1, line width=1mm,  minimum width = 0.5cm, minimum height=0.2cm,fill=black]{};

\node(out1) at (2.5, 2)[rectangle, draw, inner sep=1, line width=1mm,  minimum width = 2.5cm, minimum height=1cm,fill=black!20!white]{Pipe 1};
\node(out2) at (2.5, 0)[rectangle, draw, inner sep=1, line width=1mm,  minimum width = 2.5cm, minimum height=1cm,fill=black!20!white]{Pipe $k$};
\node(out3) at (2.5, -2)[rectangle, draw, inner sep=1, line width=1mm,  minimum width = 2.5cm, minimum height=1cm,fill=black!20!white]{Pipe $M$};
\node(out4) at (2.5, -3)[align=center]{\textbf{Egress}\\\textbf{Pipeline}};
\node(in4) at (4.3, -3)[align=center]{\textbf{Egress}\\\textbf{port}};

\draw[-, dashed](1.25, -2) -- (1.25, -3.5);
\draw[-, dashed](3.75, -2) -- (3.75, -3.5);

\node(oc1)  at(2.5, 1) {\bm{$\cdots$}};
\node(oc2)  at(2.5, -1) {\bm{$\cdots$}};

\node(op1) at (4, 2.4)[rectangle, draw, inner sep=1, line width=1mm,  minimum width = 0.5cm, minimum height=0.2cm,fill=black]{};
\node(op2) at (4, 2)[rectangle, draw, inner sep=1, line width=1mm,  minimum width = 0.5cm, minimum height=0.2cm,fill=black]{};
\node(op3) at (4, 1.6)[rectangle, draw, inner sep=1, line width=1mm,  minimum width = 0.5cm, minimum height=0.2cm,fill=black]{};

\node(op4) at (4, 0.4)[rectangle, draw, inner sep=1, line width=1mm,  minimum width = 0.5cm, minimum height=0.2cm,fill=black]{};
\node(op5) at (4, 0)[rectangle, draw, inner sep=1, line width=1mm,  minimum width = 0.5cm, minimum height=0.2cm,fill=black]{};
\node(op6) at (4, -0.4)[rectangle, draw, inner sep=1, line width=1mm,  minimum width = 0.5cm, minimum height=0.2cm,fill=black]{};

\node(op7) at (4, -1.6)[rectangle, draw, inner sep=1, line width=1mm,  minimum width = 0.5cm, minimum height=0.2cm,fill=black]{};
\node(op8) at (4, -2)[rectangle, draw, inner sep=1, line width=1mm,  minimum width = 0.5cm, minimum height=0.2cm,fill=black]{};
\node(op9) at (4, -2.4)[rectangle, draw, inner sep=1, line width=1mm,  minimum width = 0.5cm, minimum height=0.2cm,fill=black]{};

\node at (-4.8, 1.6){\textbf{pkt}};
\draw[-, blue, line width=0.5mm](-4.5, 1.6) -- ( 0, 1.6);
\draw[-, blue, line width=0.5mm](0, 1.6) -- (1.2, -0.4);
\draw[->, blue, line width=0.5mm](1.2, -0.4) -- (4.5, -0.4);

\draw[thick,decorate,decoration={calligraphic brace,amplitude=5pt}, line width=1.25pt](-3.75,2.75) -- (-1.25,2.75);

\draw[thick,decorate,decoration={calligraphic brace,amplitude=5pt}, line width=1.25pt](5.25,2.75) -- (12.75,2.75);

\draw[-, line width=0.5mm](-2.5, 3.75) -- (9, 3.75);
\draw[-, line width=0.5mm](-2.5, 3.75) -- (-2.5, 2.85);
\draw[-, line width=0.5mm](9, 3.75) -- (9, 2.85);

\node(header) at (6, 0)[rectangle, draw, dashed, inner sep=1,  minimum width = 1.25cm, minimum height=5cm,fill=black!20!white, align=center]{\textbf{PHV:} \\  \textbf{Headers} \\ \textbf{Metadata}
};

\node(stage1) at (7.75, 0)[rectangle, draw, inner sep=1,  line width=1mm, minimum width = 1.25cm, minimum height=5cm, align=center]{\textbf{Match} \\ \textbf{Action}\\ \textbf{Table}
};

\node(stage2) at (9.5, 1.25)[rectangle, draw, inner sep=1,  line width=1mm, minimum width = 1.25cm, minimum height=2.3cm, align=center]{\textbf{Counter} \\ \textbf{Meter}\\ \textbf{Register}
};
\node(stage22) at (9.5, -1.25)[rectangle, draw, inner sep=1,  line width=1mm, minimum width = 1.25cm, minimum height=2.3cm, align=center]{\textbf{ALU} 
};
\node(stagep) at (10.75, 0){\bm{$\cdots$}};
\node(stage3) at (12, 0)[rectangle, draw, inner sep=1,  line width=1mm, minimum width = 1.25cm, minimum height=5cm, align=center]{\textbf{Match} \\ \textbf{Action}\\ \textbf{Table}
};

\node(label1) at (7.75, -3)[align=center]{\textbf{Stage 1}};
\node(label1) at (9.5, -3)[align=center]{\textbf{Stage 2}};
\node(label1) at (12, -3)[align=center]{\textbf{Stage N}};
\draw[->,  line width=0.5mm](6, -3.5) -- (12, -3.5);

\end{tikzpicture}
\end{adjustbox}
\caption{Tofino-based switch data plane architecture}
\label{fig:arch}
\end{figure}
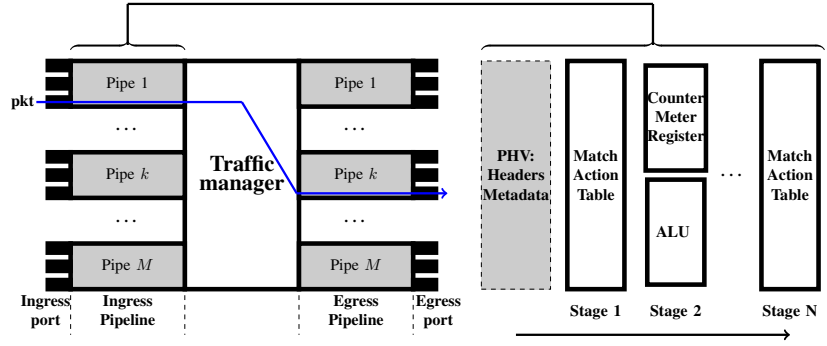
\figurename~\ref{fig:arch} shows the architecture of programmable commodity switches with Tofino ASIC (Application-Specific Integrated Circuit).
There are multiple \emph{pipes} (each composed by an \emph{ingress} and \emph{egress pipeline}) in the switch, and several ports are associated with one ingress or egress pipeline. 
Programs executed in different pipes are independent in terms of memory and computational resources.
For instance, register values in pipe 1 cannot be read by programs running in pipe 2.
When packets enter the switch, they are first processed by the input port's ingress pipeline, then, after crossing the switching matrix, are processed again by the output port's egress pipeline.
Each pipeline contains a limited number of \emph{stages}, each including one or more code \emph{blocks} (i.e., sets of operations) that are applied to each packet in sequence.
Boundaries on the number of operations executed within each stage exist to ensure that the ASIC can process packets at line rate, irrespective of the custom logic being implemented.
Blocks can contain operations including \emph{(i.)} applying match-action tables for customized packet processing, \emph{(ii.)} reading or writing counters/meters/registers for counting or storing packets, and \emph{(iii.)} calls to Arithmetic Logic Units (ALUs) for local computations. 
Note that each stage has limited hardware resources, such as memory size (Static and Ternary Random Access Memory, i.e., SRAM and TCAM) and number of ALUs.
Our implementation of INDDoS fits entirely in the ingress pipeline, leaving the egress pipeline free for other uses as discussed in Section \ref{subsec:mitigation}.

\subsection{Implementation of pairwise independent hash functions}
\begin{table}[t]
\caption{Properties of hash functions}
\label{table:hash}
\centering
\begin{adjustbox}{width=0.45\textwidth}
\begin{tabular}{|c|c|c|c|c|}
\hline
\textbf{Hash function name} & \textbf{poly} & \textbf{Reversed} & \textbf{init} & \textbf{xor}\\
\hline
CRC32 & 0x104C11DB7 & True & 0 & 0xFFFFFFFF\\
\hline
CRC32c  & 0x11EDC6F41 & True & 0 & 0xFFFFFFFF\\
\hline
CRC32d & 0x1A833982B & True & 0 & 0xFFFFFFFF\\
\hline
CRC32q & 0x1814141AB & False & 0 & 0\\
\hline
CRC32mpeg & 0x104C11DB7 & False & 0xFFFFFFFF & 0\\
\hline
\end{tabular}
\end{adjustbox}
\end{table}	
Since our P4-enabled commodity switch does not natively provide support for pairwise independent hash functions, 
we varied the usual \emph{poly}, \emph{reverse}, \emph{initial value}, and \emph{xor} parameters in the available embedded CRC32 function to produce a set of suitable pairwise independent hash functions, such as CRC32c, CRC32d, CRC32q and CRC32mpeg (see Table~\ref{table:hash}). 
These hash functions guarantee that the hashed values of the same flow key are independent.

\subsection{BACON Sketch - Implementation of updates}
We implemented BACON Sketch with $d$ rows of $w \times m$-sized registers. 
Considering that we cannot use more than two arithmetic operations to compute the input of a register (due to the atomic action rule \cite{qian2019flexgate}), we cannot calculate the $index$ of the cell to be updated as $(h_{cm}^{i}(dst)\%w) \cdot m + bucket$ (see Alg.~\ref{alg:bacon}).
Thus, we express $index$ by concatenating $bucket$ and $h_{cm}^{i}$ as a binary number (see \figurename~\ref{fig:baconT}), where the least significant bits $index[0:\log_{2}(m)-1]$ represent $bucket$, while the rest (i.e., $index[\log_{2}(m): \log_{2}(m) + \log_{2}(w) -1]$) represent $h_{cm}^{i}(dst)$.
In order to assure that $\log_{2}(w)$ and $\log_{2}(m)$ are integers, $w$ and $m$ must be powers of two.

The maximum size of registers in the  switch is $2^{17}$, which means that if $\log_{2}(m) + \log_{2}(w)$ is larger than 17, the concatenated $index$ is truncated to 17 bits. 
In order to overcome this limitation, it is possible, as shown in \figurename~\ref{fig:baconT}, to initialize up to $\log_{2}(m) + \log_{2}(w) - 17$ registers in the switch, labeled from 0 to $\log_{2}(m) + \log_{2}(w) - 18$. 
The most significant bits of $index$ are the label of the registers that should be updated. 
For instance, if $index[17:] = k$, the index $index[0:16]$ at the register labeled $k$ will be updated to 1. 
Unfortunately, this approach requires few additional pipeline stages in the switch due to using more registers.

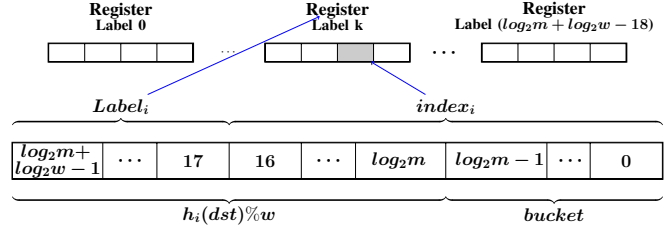
\begin {figure}[t]
\centering
\scalebox{0.48}{
\begin{tikzpicture}[]
\node (grey) [fill=black!20!white, minimum width= 1cm, minimum height=0.5cm] at (2, 0.5) {};
\node(reg1) at (1.5, 0.5)[rectangle, draw, inner sep=0]{
\myfourgrids
};
\node(reg1) at (-4.5, 0.5)[rectangle, draw, inner sep=0]{
\myfourgrids
};
\node (kk) at (-1.5, 0.5) [minimum width=2cm,minimum height=1cm, align=center] {$\cdots$};
\node (reg_label0) at (-4.5, 1.5) [align=center] {\Large{\textbf{Register}} \\ \large{\textbf{Label 0}}};

\node (reg_label1) at (1.5, 1.5) [align=center] {\Large{\textbf{Register}} \\ \large{\textbf{Label k}}};
\node (kk) at (4.5, 0.5) [minimum width=2cm,minimum height=1cm, align=center] {\Large{$\bm{\cdots}$}};
\node (reg_label2) at (7.5, 1.5) [align=center] {\Large{\textbf{Register}} \\ \large{\textbf{Label ($\bm{log_{2}{m}+log_{2}{w}-18}$)}}};

\node(reg2) at (7.5, 0.5)[rectangle, draw, inner sep=0]{
\myfourgrids
};
\node(hash) at (1.5, -2.5)[rectangle, draw, inner sep=0]{
\mytwogrids
};
\node (s1) at (9.5, -2.5) [ align=center] {\Large{$\bm{0}$}};
\node (s2) at (8, -2.5) [ align=center] {\Large{$\bm{\cdots}$}};
\draw[-](7.3, -3) -- (7.3, -2);
\node (s3) at (6, -2.5) [ align=center] {\Large{$\bm{log_{2}{m} - 1}$}};
\node (s4) at (3.2, -2.5) [ align=center] {\Large{$\bm{log_{2}{m}}$}};
\draw[-](2, -3) -- (2, -2);
\node (s5) at (1.3, -2.5) [ align=center] {\Large{$\bm{\cdots}$}};
\node (s6) at (-0.5, -2.5) [ align=center] {\Large{$\bm{16}$}};
\node (s7) at (-2.5, -2.5) [ align=center] {\Large{$\bm{17}$}};
\draw[-](-1.5, -3) -- (-1.5, -2);
\node (s8) at (-4.2, -2.5) [ align=center] {\Large{$\bm{\cdots}$}};
\node (s8) at (-6.3, -2.5) [ align=center] {\Large{$\bm{log_{2}{m} + }$}\\\Large{ \bm{$log_{2}{w}- 1$}}};
\draw[-](-5, -3) -- (-5, -2);

\draw[thick,decorate,decoration={calligraphic brace,amplitude=5pt}, line width=1.25pt](10.5,-3.5) -- (4.5,-3.5);
\node(bucket) at (7.5,-4){\Large{\bm{$bucket$}}};

\draw[thick,decorate,decoration={calligraphic brace,amplitude=5pt}, line width=1.25pt](4.5,-3.5) -- (-7.5,-3.5);
\node(bucket) at (-1.5,-4){\Large{$\bm{h_{i}(dst)\%w}$}};
\draw[thick,decorate,decoration={calligraphic brace,amplitude=5pt, mirror, reverse path}, line width=1.25pt](-1.5,-1.5) -- (10.5,-1.5) ;
\node(index) at (4.5,-1){\Large{$\bm{index_{i}}$}};
\draw[->, blue](index) -- (grey);
\draw[thick,decorate,decoration={calligraphic brace,amplitude=5pt, mirror, reverse path}, line width=1.25pt](-7.5,-1.5) -- (-1.5,-1.5) ;
\node(label) at (-4.5,-1){\Large{$\bm{Label_{i}}$}};

\draw [->, blue](label) -- (1, 1.5);

\end{tikzpicture}
}
\caption{BACON Sketch updates to overcome the switch's hardware limitations}
\label{fig:baconT}
\end{figure}
\subsection{BACON Sketch - Implementation of queries}
Calculating the sum of values in a Bitmap register is a costly operation inside a switch, since it needs to iteratively load and compute the sum of many individual values. 
For this reason, we used an \emph{auxiliary} Count-min Sketch to directly store the updated sum of each Bitmap register (see Fig. \ref{DDoS}). 
This Count-min Sketch is composed of $d \times w$ counters, each associated with the number of 1s in the corresponding Bitmap register of BACON Sketch. 
This sketch is updated as follows: when a packet arrives, the $h_{cm}^{i}(dst)$-th counter may \emph{(i.)} stay the same if $index_{i}$ in row $i$ of BACON Sketch is already 1 or \emph{(ii.)} increase by 1 if the same index of BACON Sketch is 0. 
The updated value $E_{i}$ in the counter $index_{i}$ is the same as the queried value in row $i$ of BACON Sketch. 
Finally, the flow cardinality of $dst$ is obtained by using \emph{if-else statements} to take the minimum of estimated cardinalities $E_{i}$ in each row of the auxiliary sketch. 

\subsection{Implementation of INDDoS}
If a queried flow cardinality $\hat{E}_{dst}$ is equal to $\theta n + 1$ for one or more destinations, the switch's data plane prepares a short \emph{digest} packing the 4 bytes of the detected destination IP(s), and sends it to the controller.
The equality prevents generating multiple digests for the same $dst$ in one time interval.
At the end of a time interval the switch resets all registers and starts a new round of DDoS victim identification.

\subsection{Limitations hindering the implementation of other solutions}
In this subsection, we briefly explain why competing state-of-the-art sketch-based per-flow cardinality estimation algorithms cannot currently be implemented in P4-enabled programmable hardware switches:
\subsubsection{Virtual HyperLogLog (vHLL) \cite{xiao2017cardinality}}
we were not able to implement this strategy for three reasons: \emph{(i.)} we could not find a way to count tailing 0s for the Update operation in HyperLogLog; \emph{(ii.)} 
the required harmonic mean for querying cardinality is currently not implementable since P4 does not support the calculation of inverse numbers;
\emph{(iii.)} 
vHLL needs to access several register cells to query the cardinality of a single flow but, in our hardware, the same register cannot be accessed more than once per packet.
\subsubsection{SpreadSketch \cite{tangspreadsketch} (combination of Multiresolution Bitmap \cite{estan2003bitmap} and Count-min Sketch \cite{cormode2005improved})}
we were not able to implement this strategy for two reasons: \emph{(i.)} P4 does not support the computation of logarithms, which is necessary for Multiresolution Bitmap to estimate the cardinality (we proposed a logarithm estimation algorithm in P4 in \cite{ding2019estimating}, but currently it can only be implemented in Behavioral model and not in hardware); \emph{(ii.)} arithmetic operations on the same metadata, in our hardware, cannot occur more than twice, while more are required here.

Thus, unfortunately, we cannot perform a direct comparison of our strategy with the state of the art since the solutions proposed in literature do not meet the stringent hardware constraints of current Tofino-based programmable switches, and therefore they cannot be entirely implemented in the switch pipeline. 
We are aware that, in general terms, BACON Sketch is a rather simple solution for per-flow cardinality estimation and more refined strategies have already been proposed. However, as pointed out earlier, BACON is the only solution that can be implemented in existing commodity programmable switches.


\section{Potential integration of INDDOS in a full DDoS defense mechanism} \label{sec:integration}
In this section we explain how a DDoS defense system could benefit from the early detection of potential DDoS traffic in the data plane, as performed by INDDoS, as well as how the latter could benefit from the supervision of the former.

\subsection{Setting the DDoS detection threshold}
The performance of INDDoS depends on selecting an appropriate threshold fraction $\theta$, which determines the minimum number of sources contacting the same destination that trigger an alarm, discriminating between legitimate and abnormal traffic.
In principle, to strike a good balance between false positives and false negatives, $\theta$ should be chosen so that $\theta n$  \emph{(i.)} is larger than the largest flow cardinality of any host in normal conditions (i.e., not under attack), and \emph{(ii.)} is not so large that it fails to detect some attacks. Additionally, the threshold should be dynamically adapted to traffic changes.

Even though the definition of a detailed strategy for dynamic threshold setting is beyond the scope of this paper, we provide some indications on how it could be designed. In the simplest case, it could be set by mirroring and analyzing a large trace of legitimate traffic sampled in each switch. Alternatively, considering a more complex DDoS detection function at the controller, it could be set via a slow negative feedback control loop: start with a relatively high threshold, and slowly decrease it over time; at some point, the switch will start flagging some destinations as potentially under attack, but the controller could determine whether those are false positives, and if so raise the threshold back up.
Finally, when $\theta$ changes, the BACON sketch size should be tweaked according to our theoretical analysis, in order to maximize the detection performance of INDDoS.
In different deployment scenarios that attacks can traverse multiple disjoint paths, a controller could implement a network-wide strategy tweaking $\theta$ on the switches accordingly.

\subsection{DDoS attack mitigation inside programmable switches} \label{sec:ddos_mitigation}
INDDoS enables the early detection of DDoS attack victims in the data plane pipeline that can be exploited to immediately mitigate them.
Once INDDoS has identified a possible DDoS victim and sent a digest to the controller, two simple DDoS victim-based mitigation operations could be triggered in the data plane, i.e., \emph{Drop} and \emph{Rate limit}:
\begin{itemize}
    \item \textbf{Drop}: 
    the controller automatically configures, in all programmable switches in the network, the relevant egress Access Control List (ACL) with identified attack traffic features (e.g. source, destination IPs, ports) in order to drop the packets belonging to likely malicious flows.
    \item \textbf{Rate limit}: a meter is used to aggregate traffic towards the identified victim. If the rate is larger than a pre-defined threshold, the traffic rate is limited.
\end{itemize}

Although we already implemented both operations in P4, we are keenly aware that neither can actually implement proper DDoS mitigation alone.
For that, external aid from the controller or an additional, more precise and computationally expensive detection mechanism (operating on the limited subset of suspicious flows identified by INDDoS) is needed, as outlined in the next subsection.

\subsection{DDoS attack mitigation outside programmable switches} \label{subsec:mitigation}
More refined and complex strategies, making use of external servers, have been proposed in literature for DDoS attack mitigation.
Here we describe a few and explain how they can benefit from INDDoS.

Bohatei \cite{fayaz2015bohatei} is a strategy that spins up virtual machines (VMs) in servers for the identification, analysis and mitigation of suspicious traffic as needed: once suspicious traffic is detected, a resource manager determines the type, number, and location of VMs to be instantiated, and such traffic is steered to them for further analysis (and possibly mitigation) with minimal impact on users' perceived latency.
Another approach, named Poseidon \cite{zhang2020poseidon}, performs DDoS mitigation by combining the capabilities of programmable switches and of sets of external servers. The additional mitigation functionalities provided by the servers, with respect to those provided by programmable switches, significantly improves the performance with respect to DDoS defense: Poseidon can mitigate sixteen different types of attacks exploiting different protocols (i.e., ICMP, TCP, UDP, HTTP). 

However, both Bohatei and Poseidon consider a ``DDoS-defense-as-a-service scenario", that is, they assume that the DDoS victim is known a-priori. INDDoS is complementary to these approaches since it can effectively detect and identify likely DDoS victims, and hence route only a manageable subset of network traffic towards these detection services.






\section{Experimental evaluation} \label{sec:evaluation}
\begin {figure}[t]
\centering
\scalebox{0.7}{
\begin{tikzpicture}

\node[inner sep=0pt] (switch1) at (2,0){\includegraphics[scale = 0.5]{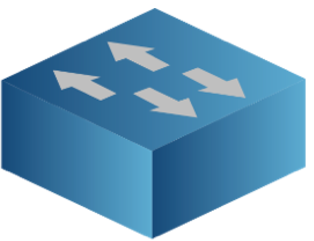}};
\node[inner sep=0pt, align=center] (switch-name) at (2,-1.5){\large{\textbf{Ports: 33 x 100Gbps}} \\ \large{\textbf{Programmable commodity}} \\\large{\textbf{switch  with Tofino ASIC}}};
\calloutquote[width=0.5*\linewidth,position={(2,-1)},fill=orange,rounded corners,align=center]{\large{\textbf{INDDoS}}}

\node[inner sep=0pt] (cmd) at (4.5,2){\includegraphics[scale = 0.6]{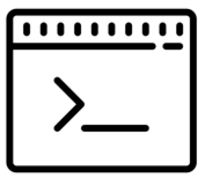}};
\node[inner sep=0pt, align=center] at (3,2.5){\large{\textbf{Digest}}};
\draw[->, thick, darkgreen](switch1)|-(cmd);
\node[inner sep=0pt, align=center] (switch-name) at (6.2,2){\large{\textbf{Command}}\\\large{\textbf{line}}\\ \large{\textbf{interface}}};

\node[inner sep=0pt] (server1) at (-1.5,-1){\includegraphics[scale = 0.5]{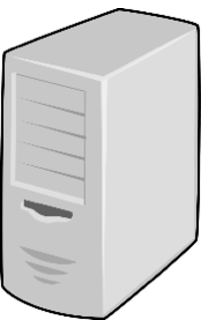}};
\node[inner sep=0pt, align=center] (switch-name) at (-2,-2){\large{\textbf{Traffic sender}}};

\node[inner sep=0pt] (server2) at (5.5,-1){\includegraphics[scale = 0.5]{imgs/server.eps}};
\node[inner sep=0pt, align=center] (switch-name) at (6,-2){\large{\textbf{Traffic receiver}}};

\draw[thick,-](-2, -0.6) edge[bend left=140, draw=black] node [left] {}  (1.2,0);
\draw[thick,-](6, -0.6) edge[bend right=150, draw=black] node [left] {}  (2.8,0);

\draw[thick, -](-2,0) edge[bend left=20, draw=orange,->] node [left] {}  (1.5,0.5);

\draw[thick, -](2.5, 0.5) edge[bend left=20, draw=orange,->] node [left] {}  (6.5,-0.2);
\node[inner sep=0pt, align=center] (traffic) at (-0.5,1){\large{\textbf{Traffic}}};
\node[inner sep=0pt, align=center] (traffic) at (4.5,1){\large{\textbf{Traffic}}};

\filldraw[orange] (-2,-0.6) circle (2pt);
\node[inner sep=0pt, align=center] (port1) at (-3.2,-0.6){\large{\textbf{Port:}} \\ \large{\textbf{10Gbps}}};

\filldraw[orange] (6,-0.6) circle (2pt); 
\node[inner sep=0pt, align=center] (port2) at (7.2,-0.6){\large{\textbf{Port:}} \\ \large{\textbf{10Gbps}}};

\filldraw[orange] (1.2,0) circle (2pt); 
\filldraw[orange] (2.8,0) circle (2pt); 

\end{tikzpicture}
}
\caption{The physical testbed for the experiments on DDoS victim identification}
\label{fig:testbed}
\end{figure}
We implemented INDDoS in a commodity Edgecore Wedge-100BF-32X switch equipped with Barefoot Tofino 3.3 Tbps ASIC \cite{Tofino}, which supports up to 32 100\,Gbps ports. 
Due to the high cost of 100\,Gbps interfaces, we connected the switch to two servers (Intel(R) Xeon(R) CPU E3-1220 V2 @ 3.10GHz, 16\,GB RAM) using 10\,Gbps Ethernet interfaces, as shown in \figurename~\ref{fig:testbed}.
We also implemented a P4-based \emph{simple forwarding} strategy for comparison purposes. 
Packets are sent across the switch via Tcpreplay \cite{tcpreplay}. 

\subsection{Evaluation metrics and settings}

\subsubsection{Testing flow traces}
In our first experiment (\textit{Exp 1}) we used, as done in previous works on DDoS detection \cite{huang2017sketchvisor}\cite{electronics9030413}, a 50s passive CAIDA flow trace \cite{caida} collected from a 10Gbps backbone link, which we divided into 10 time intervals. 
Each 5s time interval contains around 2.3 million packets and 60 thousand distinct source IPs.
Note that, although the trace did not contain actual DDoS traffic, we were still able (by using a low, sensitive threshold) to successfully detect outliers, despite the number of incoming connections for such cases being far smaller than in actual volumetric DDoS scenarios.
In other words, our experiment is a pressure test: if INDDoS can detect such outliers accurately, by properly setting a threshold between legitimate and DDoS traffic, it should have little trouble in properly distinguishing actual attack traffic where such difference is more marked.

In a second experiment (\textit{Exp 2}) we then considered real DDoS attacks, namely from Booter \cite{santannajjIM2015}.
Booter is a class of on-demand services that provide illegal support to launch DDoS attacks targeting websites and networks.
We considered four 50s Booter DNS-amplification DDoS attack traces (i.e. booter 1, 4, 6, 7), which are the ones with the highest number of attack source IPs among all the available traces in \cite{santannajjIM2015}.
This allows us to investigate the other extreme with respect to just considering the CAIDA trace without attacks (Exp 1). 
Table~\ref{table:flow} reports the salient properties of the chosen traces. 
We split each of them into 10 time intervals as for the CAIDA trace, and appended each 5s attack trace at the end of its respective legitimate trace.
The attack targets are the destinations specified in the Booter traces (one per trace), and the attacked destination IPs are not present in the CAIDA normal traffic.
In this way we generated four traces containing a single attack, as well as a fifth \emph{Mixed} trace containing all four attacks simultaneously along with the legitimate CAIDA traffic.
Note that in typical DNS-amplification DDoS attacks the number of usable amplifying reflectors is smaller than that of bots and usable reflectors: this means that we are considering the lowest number of source IPs that might be expected from a volumetric (non-semantic) attack, and thus we are not posing ourselves in the most favorable condition (i.e., with the strongest possible outliers).

\begin{table}[t]
\caption{Properties of DDoS flow traces \cite{santannajjIM2015}}
\label{table:flow}
\centering
\begin{adjustbox}{width=0.4\textwidth}
\begin{tabular}{|c|c|c|}
\hline
\textbf{DDoS trace name} & \textbf{Packets per second} & \textbf{Attack source IPs}\\
\hline
Booter 6 & $\sim$ 90000 & 7379\\
\hline
Booter 7  & $\sim$ 41000 & 6075\\
\hline
Booter 1 & $\sim$ 96000 & 4486\\
\hline
Booter 4 & $\sim$ 80000 & 2970\\
\hline
\end{tabular}
\end{adjustbox}
\end{table}	

\subsubsection{Metrics}
We chose \emph{recall} $R$ and \emph{precision} $Pr$ as the key metrics to evaluate INDDoS, defined as follows:
  \vspace{0pt}
  \begin{align*}
	R &= \frac{Count_{DDoS victim}^{detected/true}}{Count_{DDoS victim}^{detected/true} +Count_{DDoS victim}^{undetected/true}}\\
		Pr &= \frac{Count_{DDoS victim}^{detected/true}}{Count_{DDoS victim}^{detected/true} +Count_{DDoS victim}^{detected/false}} 
\end{align*}
The count of DDoS victims is obtained via the command line interface provided by our Tofino-based switch.

In our evaluations, we considered the well known \emph{F1 score} ($F1$) as a compact metric incorporating both precision and recall, and thus measuring the accuracy of our strategy.
It is defined as:
\begin{equation*}
    F1 = \frac{2 \cdot Pr \cdot R}{Pr + R}
\end{equation*}
All results are mean values computed across 10 time intervals. 

\subsubsection{Tuning parameters}
The default BACON Sketch size $d \times w \times m$ is $3\times 1024\times 1024$, and the DDoS threshold fraction $\theta$ is set to 0.5\% as per \cite{huang2017sketchvisor} in the first experiment, and to 1\% in the second, to better separate legitimate and DDoS traffic.

\begin{table}[t]
\caption{Comparison of INDDoS performance as a function of BACON Sketch parameters and with Spread Sketch \cite{tangspreadsketch}}
\label{table:stage}
\centering
\begin{adjustbox}{width=0.4\textwidth}
\begin{tabular}{|c|c|c|c|}
\hline
\textbf{BACON Sketch size ($d \times w \times m$)} & \textbf{Recall} &\textbf{Precision} & \textbf{F1 score}\\
\hline
$3 \times 1024 \times 1024$ & 0.96 & 0.99 & 0.97 \\
\hline
\hline
$1 \times 2048 \times 1024$ & 0.98 & 0.54 & 0.70 \\
\hline
$1 \times 1024 \times 2048$ & 0.94 & 0.38 & 0.54 \\
\hline
$5 \times  1024 \times 512$ & 0.12 & 1.0 & 0.22 \\
\hline
$5 \times 512 \times 1024$ & 0.96 & 0.89 & 0.92 \\
\hline
\hline
\textbf{Spread Sketch size ($d \times w \times m$)} & \textbf{Recall} &\textbf{Precision} & \textbf{F1 score}\\
\hline
$3 \times 1024 \times 1024$ & 0.92 & 0.94 & 0.93 \\
\hline
\end{tabular}
\end{adjustbox}
\end{table}	

\begin{figure*}[t]
  \centering

\subfigure[Sensitivity analysis to $d$]{
{
\resizebox{0.3\linewidth}{0.205\textwidth}{
\begin{tikzpicture}[font=\large]

\definecolor{color0}{rgb}{0.12156862745098,0.466666666666667,0.705882352941177}
\definecolor{color1}{rgb}{1,0.498039215686275,0.0549019607843137}
\definecolor{color2}{rgb}{0.172549019607843,0.627450980392157,0.172549019607843}

\begin{axis}[
legend cell align={left},
legend style={fill opacity=0.8, draw opacity=1, text opacity=1, at={(0.97,0.03)}, anchor=south east, draw=white!80.0!black},
tick align=outside,
tick pos=left,
x grid style={white!69.01960784313725!black},
xlabel={Number of hash functions \(\displaystyle d\)},
xmajorgrids,
xmin=0.9, xmax=3.1,
xtick style={color=black},
xtick={1, 2, 3},
y grid style={white!69.01960784313725!black},
ymajorgrids,
ymin=-0.0495, ymax=1.0395,
ytick style={color=black}
]
\addplot [semithick, color0, mark=*, mark size=3, mark options={solid,fill=white,draw=red}]
table {%
1 0.96
2 0.96
3 0.96
};
\addlegendentry{Recall}
\addplot [semithick, color1, mark=x, mark size=3, mark options={solid,fill=white,draw=red}]
table {%
1 0.35
2 0.92
3 0.99
};
\addlegendentry{Precision}
\addplot [semithick, color2, mark=triangle*, mark size=3, mark options={solid,fill=white,draw=red}]
table {%
1 0.51
2 0.95
3 0.97
};
\addlegendentry{F1 score}
\end{axis}

\end{tikzpicture}
    \label{fig:sub:d}
    }
    }
    }
\subfigure[Sensitivity analysis to $w$]{
{\resizebox{0.3\linewidth}{0.205\textwidth}{
\begin{tikzpicture}[font=\large]

\definecolor{color0}{rgb}{0.12156862745098,0.466666666666667,0.705882352941177}
\definecolor{color1}{rgb}{1,0.498039215686275,0.0549019607843137}
\definecolor{color2}{rgb}{0.172549019607843,0.627450980392157,0.172549019607843}

\begin{axis}[
legend cell align={left},
legend style={fill opacity=0.8, draw opacity=1, text opacity=1, at={(0.97,0.03)}, anchor=south east, draw=white!80.0!black},
tick align=outside,
tick pos=left,
x grid style={white!69.01960784313725!black},
xlabel={Output size \(\displaystyle w\)},
xmajorgrids,
xmin=83.2, xmax=1068.8,
xtick style={color=black},
xtick={128, 256, 512, 1024},
y grid style={white!69.01960784313725!black},
ymajorgrids,
ymin=-0.0495, ymax=1.0395,
ytick style={color=black}
]
\addplot [semithick, color0, mark=*, mark size=3, mark options={solid,fill=white,draw=red}]
table {%
128 0.75
256 0.88
512 0.96
1024 0.96
};
\addlegendentry{Recall}
\addplot [semithick, color1, mark=x, mark size=3, mark options={solid,fill=white,draw=red}]
table {%
128 0.02
256 0.02
512 0.79
1024 0.99
};
\addlegendentry{Precision}
\addplot [semithick, color2, mark=triangle*, mark size=3, mark options={solid,fill=white,draw=red}]
table {%
128 0.04
256 0.05
512 0.86
1024 0.97
};
\addlegendentry{F1 score}
\end{axis}

\end{tikzpicture}
    \label{fig:sub:w}
    }
    }
    }
\subfigure[Sensitivity analysis to $m$]{
{\resizebox{0.3\linewidth}{0.205\textwidth}{
\begin{tikzpicture}[font=\large]

\definecolor{color0}{rgb}{0.12156862745098,0.466666666666667,0.705882352941177}
\definecolor{color1}{rgb}{1,0.498039215686275,0.0549019607843137}
\definecolor{color2}{rgb}{0.172549019607843,0.627450980392157,0.172549019607843}

\begin{axis}[
legend cell align={left},
legend style={fill opacity=0.8, draw opacity=1, text opacity=1, at={(0.97,0.03)}, anchor=south east, draw=white!80.0!black},
tick align=outside,
tick pos=left,
x grid style={white!69.01960784313725!black},
xlabel={Bitmap register size \(\displaystyle m\)},
xmajorgrids,
xmin=83.2, xmax=1068.8,
xtick style={color=black},
xtick={128, 256, 512, 1024},
y grid style={white!69.01960784313725!black},
ymajorgrids,
ymin=-0.0495, ymax=1.0395,
ytick style={color=black}
]
\addplot [semithick, color0, mark=*, mark size=3, mark options={solid,fill=white,draw=red}]
table {%
128 0
256 0
512 0.17
1024 0.96
};
\addlegendentry{Recall}
\addplot [semithick, color1, mark=x, mark size=3, mark options={solid,fill=white,draw=red}]
table {%
128 0
256 0
512 0.96
1024 0.99
};
\addlegendentry{Precision}
\addplot [semithick, color2, mark=triangle*, mark size=3, mark options={solid,fill=white,draw=red}]
table {%
128 0
256 0
512 0.28
1024 0.97
};
\addlegendentry{F1 score}
\end{axis}

\end{tikzpicture}
    \label{fig:sub:m}
    }
    }
    }
\caption{Sensitivity analysis of INDDoS to the parameters of BACON Sketch}
  \label{fig:sens}
\end{figure*}
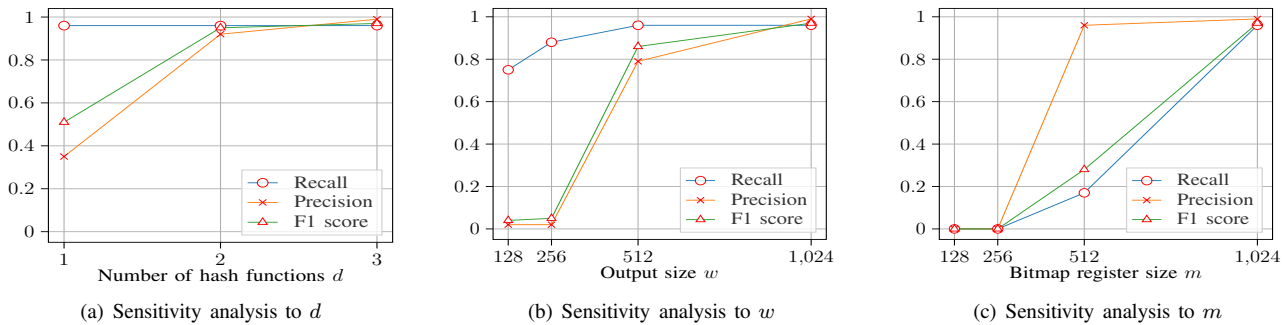

\subsection{Exp 1: evaluation of DDoS victim identification accuracy} \label{sec:evaluation_victim}
According to Theorems 1 and 2,  the Bitmap register size $m$ impacts both upper and lower bounds on the quality of BACON Sketch's estimation: increasing $m$ increases the distance between the upper bound of the estimation and the true flow cardinality, but also reduces that between the lower bound and the true value.
This implies that $m$ should be neither too large nor too small.
In our flow trace, the largest flow cardinality in each time interval is within $[2^9, 2^{10}]$, so we set $m$ to 1024. 
Conversely, larger output sizes $w$ generate lower false negative bounds and do not affect the false positive bound, but in order to apply the false positive bound proved in Theorem~\ref{th:4}, $w$ should satisfy $w \ge \frac{4}{\theta}$.
Therefore, given $\theta = 0.5\%$ in our case, $w$ should be larger than $\frac{4}{0.5\%} = 800$.
However, due to hardware limitations, we cannot use arbitrarily large values for $w$; in fact, the only feasible parameter combinations that satisfy all conditions on $m$ and $w$, while not exceeding the available pipeline stages, are: $3 \times 1024 \times 1024$ and $1 \times 2048 \times 1024$. 
Considering that $1 \times 2048 \times 1024$ only uses 1 hash function,  the estimated flow cardinality has a high chance (i.e., $\frac{1}{2}$) to be out of both upper bound and lower bound. 
Note that the number of hash functions $d$ = 2 was not considered since 2 hash functions support the same largest register size (i.e., $w \times m$ = $1024 \times 1024$) as 3 functions within the available pipeline stages. 
Likewise, 4 hash functions support the same largest register size (i.e., $w \times m$ = $512 \times 1024$) as 5 functions, but neither can satisfy Theorem 4. 
Thus, according to the theoretical analysis, the best choice of the sketch size should be $3 \times 1024 \times 1024$ (hence our default choice).

Table~\ref{table:stage} shows the performance of INDDoS for different BACON Sketch sizes. 
In all cases listed in the table, all available pipeline stages are required to implement the strategy, while the occupied memory is much lower than available switch's memory. 
Results show good performance on both recall and precision for  $3 \times 1024 \times 1024$, which translates to a high F1 score. 
With only one hash function in BACON Sketch, INDDoS shows high recall, i.e., it is able to identify most DDoS victims, but low precision, meaning that a lot of false identifications take place. 
When instead $m=2048$, the recall of $1 \times 1024 \times 2048$ is as high as for $1 \times 2048 \times 1024$, but its precision decreases by a third. 
In contrast, when $m=512$ and the sketch size is $5 \times  1024 \times 512$, the recall significantly decreases leading to a much lower F1 score than for $5 \times 512 \times 1024$. 
These two results indicate that a larger $m$ causes higher recall but lower precision, and vice versa.
The recall for sketch size $3 \times 1024 \times 1024$ is comparable to the one for size $5 \times 512 \times 1024$,  but the precision is 10 percentage points higher. 
Note that in higher-speed networks, where the largest flow cardinality during a time interval may exceed 1024, INDDoS can still work properly (even in currently available switches) by shrinking the time interval.

\subsubsection{Comparison with the state of the art} 
We considered SpreadSketch \cite{tangspreadsketch} as an alternative sketch for INDDoS rather than BACON Sketch and compared the DDoS detection performance in terms of simulations. SpreadSketch is three-dimensional as BACON Sketch, and we set the same size of $3 \times 1024 \times 1024$ for the two approaches.
The results are shown in Table \ref{table:stage}. When BACON Sketch is adopted, both recall and precision are slightly higher, and therefore the overall detection performance (i.e., F1 score) is also better.
This is because the query operation in SpreadSketch uses Multiresoultion Bitmap for flow cardinality estimation, while BACON Sketch uses Direct Bitmap. When there are not many collisions in the sketch, the Direct Bitmap adopted by BACON Sketch leads to slightly higher accuracy than Multiresoultion Bitmap. Additionally, the great advantage of BACON Sketch is that it can be fully executed in the data plane (for both update and query operations), while the query operation of SpreadSketch can only be executed by the controller. This leads to great advantages in terms of communication overhead and detection speed, as it will be shown in Section \ref{sec:comm_overhead}.
\subsection{Exp. 1: sensitivity analysis of DDoS victim identification}
In order to show how the performance of INDDoS is sensitive to different tuning parameters (number of hash functions $d$, output size $w$ and Bitmap register size $m$) in BACON Sketch, we conducted some experiments by varying tuning parameters individually, while respecting the resource constraints of our switch. 
The results are reported in \figurename~\ref{fig:sens}.

\subsubsection{Sensitivity to number of hash functions $d$}
\figurename~\ref{fig:sub:d} shows how INDDoS behaves with $d$ ranging from 1 to 3. 
Recall is high in all three cases, but precision increases as $d$ increases. 
This is because more hash functions lead to smaller number of collisions in the Count-min Sketch side of BACON Sketch, and hence to a lower overestimation of the number of source IPs contacting a specific destination. 
\subsubsection{Sensitivity to output size of hash functions $w$}
\figurename~\ref{fig:sub:w} shows how INDDoS performs by varying $w$ from 128 to 1024. 
Clearly, precision is very sensitive to $w$, and it significantly decreases as $w$ decreases. 
Conversely, even for small values of $w$ (e.g., $w = 128$), recall remains above 0.75. 
\subsubsection{Sensitivity to Bitmap register size $m$}
\figurename~\ref{fig:sub:m} shows how INDDoS performance changes along with the value of $m$. 
Since the threshold in each time interval is $\theta n$, which is around $0.5\% \cdot 60000 = 300$ in our evaluation, when $m$ is below 300, i.e., $m \in \{128, 256\}$, no victim can be detected and thus the F1 score is 0.
Precision of $m = 512$ and of $m = 1024$ is comparable, but recall is much lower for $m=512$. 
As proven theoretically, a smaller $m$ causes a smaller lower bound on the flow cardinality estimation $\hat{E}_{dst}$, so it is more difficult for $\hat{E}_{dst}$ to exceed the threshold. 

\subsection{Exp. 1: communication overhead and detection speed comparison with the state of the art} \label{sec:comm_overhead}
\begin{table}
 \caption{Communication overhead comparison}
\label{comm}
 \centering
\resizebox{0.6\linewidth}{!}{
\begin{tabular}{|c|c|}
\hline
\textbf{Strategy}& \textbf{Communication overhead (per $T_{int}$)} \\
\hline
INDDoS (BACON Sketch) &  ($4 \cdot N_{victim}$)B\\
\hline
SpreadSketch ($3\times 1024 \times 1024$) & 0.375MB \\
\hline
NetFlow (sampling rate 1/2000) &  1.15MB\\
\hline
\end{tabular}
}
\end{table}
In this subsection, we compare INDDoS with two state-of-the-art approaches, i.e., SpreadSketch \cite{tangspreadsketch} and Netflow \cite{Netflow}, in terms of communication overhead and detection speed. SpreadSketch is queried by the controller, while NetFlow is used to collect packet statistics at the controller side; in both cases, the DDoS detection logic is executed by the controller.

Table \ref{comm} reports a comparison in terms of communication overhead between the considered strategies with focus on a single time interval $T_{int}$.
The sketch size of BACON Sketch and SpreadSketch is set to $3\times 1024 \times 1024$ (see Section \ref{sec:evaluation_victim}).
However, the two approaches lead to very different overheads. INDDoS only reports detected victim IP addresses to the controller: considering that each IP address is 4 bytes long, the communication overhead is $(4 \cdot N_{victim})$B, where $N_{victim}$ is the number of detected DDoS victims by INDDoS in $T_{int}$.
Instead, SpreadSketch needs to send the whole sketch to the controller, at the end of $T_{int}$, for further processing (i.e., query on the sketch) and victim identification.
As each SpreadSketch cell occupies 1 bit, the overall size of the sketch is  $\frac{3\cdot 1024 \cdot 1024}{8}$B = 0.375MB. 
Concerning NetFlow, the considered sampling rate is 1/2000, being the recommended value for 10Gbps link speed \cite{netflow-sample} and the maximum rate not affecting DDoS detection accuracy \cite{10.1145/3229598.3229605}.
The average total traffic in each 5s time interval, for the considered flow trace, is 2.3GB, meaning that the average sampled communication overhead is around $\frac{2300}{2000}$MB $=$ 1.15MB.
With respect to the two existing approaches, the communication overhead generated by INDDoS is almost negligible, given that $N_{victim}$ is expected to be on average low in any $T_{int}$. If no victim is identified (i.e., $N_{victim}=0$) no communication is needed between the programmable switch and the controller.

Concerning detection speed, INDDoS reports to the controller immediately once a victim is detected and identified (see Fig. \ref{fig:arch}). Conversely, the other two methods require the execution of the DDoS detection logic by the controller at the end of every $T_{int}$: clearly, the longer $T_{int}$ is, the faster (on average) our approach is with respect to the state of the art.

\begin{table}[t]
\caption{INDDoS performance for Booter DDoS attacks}
\label{table:booterattack}
\centering
\begin{adjustbox}{width=0.4\textwidth}
\begin{tabular}{|c|c|c|c|}
\hline
\textbf{DDoS attack flow trace} & \textbf{Recall} &\textbf{Precision} & \textbf{F1 score}\\
\hline
Booter 6 & 1.0 (1/1) & 1.0 (1/1) & 1.0 (1/1) \\
\hline
Booter 7 & 1.0 (1/1) & 1.0 (1/1) & 1.0 (1/1) \\
\hline
Booter 1 & 1.0 (1/1) & 1.0 (1/1) & 1.0 (1/1) \\
\hline
Booter 4 & 1.0 (1/1) & 1.0 (1/1) & 1.0 (1/1) \\
\hline
Mixed & 1.0 (4/4) & 1.0 (4/4) & 1.0 (4/4) \\
\hline
\end{tabular}
\end{adjustbox}
\end{table}

\subsection{Exp. 2: evaluation of DDoS victim identification accuracy under Booter DDoS attacks}
Table~\ref{table:booterattack} shows the performance of INDDoS on victim identification under actual DDoS attacks with different number of attack source IPs. 
The threshold is approximately $60000 \cdot 1\% = 600$, which is larger than the largest flow cardinality for legitimate traffic in each 5s interval.
Using these parameters, INDDoS can always correctly identify the victim even though the number of attack source IPs in all cases is greater than the Bitmap size (1024), since our threshold is lower than that and a saturated count (i.e., all 1s in the Bitmap register) still results in flagging a destination as potential DDoS victim. 
Even though there are four distinct DDoS attacks in the Mixed case, INDDoS correctly identifies all four different victims.  
This suggests that, when a suitable threshold is used (insights on how to set it are given in Section \ref{sec:integration}), INDDoS can identify victims almost perfectly. 

Furthermore, we observe that in our CAIDA flow trace the maximum flow cardinality in 5s time intervals is around 500, whereas in 1s intervals it is about 260.
Therefore, shrinking the time interval not only allows INDDoS to detect attacks on faster networks, but it also increases the difference between legitimate and DDoS traffic, simplifying the selection of a threshold and increasing detection accuracy.

\subsection{Evaluation of impact on network performance}
\begin{table}
 \caption{Network performance of INDDoS in the commodity switch}
\label{qos}
 \centering
\resizebox{0.6\linewidth}{!}{
\begin{tabular}{|c|c|c|c|c|c|}
\hline
\multirow{2}{*}{\textbf{Type}} &\multirow{2}{*}{\textbf{iPerf size}} & \multirow{2}{*}{\textbf{Throughput}}  & \multirow{2}{*}{\textbf{Jitter}} & \multirow{2}{*}{\textbf{Packet loss}}& {\textbf{Average additional}} \\
~& ~ &~ &~ &~ & \textbf{processing time w.r.t.} \\
~& ~ &~ &~ &~ & \textbf{simple forwarding} \\
\hline
TCP & 64\,KB & 9.02\,Gbps & / & / & 106\,ns \\
\hline
TCP & 128\,KB & 9.41\,Gbps  & /&/  & 101\,ns\\
\hline
UDP & 500\,B & 1.03\,Gbps  & 0.003\,ms & 0\% & 102\,ns \\
\hline
UDP & 1470\,B & 2.87\,Gbps  & 0.004\,ms&0\% & 102\,ns \\
\hline
UDP & 3000\,B & 5.45\,Gbps  & 0.004\,ms &0\%& 107\,ns  \\
\hline
UDP & 6000\,B & 9.62\,Gbps  & 0.004\,ms&0\% & 104\,ns \\
\hline
UDP & 9000\,B & 9.67\,Gbps  & 0.006\,ms&0\%  & 101\,ns \\
\hline
\end{tabular}
}
\end{table}
We used iPerf3 \cite{iperf3} to indirectly measure the performance of INDDoS in the switch. 
The results are reported in Table~\ref{qos}. 
We first generated 10\,Gbps of TCP traffic with iPerf buffer size 64\,KB and 128\,KB from one server to the other across the switch. 
In both cases, the throughput reached more than 9\,Gbps, i.e., the TCP traffic could be processed at line-rate (with 10 Gbps interfaces). 
We then generated 10\,Gbps of UDP traffic with different iPerf sizes ranging from 500B to 9000B. 
Using datagrams smaller than 1470B (i.e., the typical value for home access) our server could not reach 10Gbps, being the server's CPU processing capacity a bottleneck in this settings. Therefore, we increased the \emph{maximum transmission unit} (MTU) to allow UDP datagram sizes larger than 1470B. 
The results show that the throughput and jitter increased as packet size increased, but without packet loss. 
When the datagram size reached 6000B, the throughput was 9.62Gbps, indicating that INDDoS could also process UDP datagrams at line-rate.
Although our testbed limitations did not allow us to test the switch behaviour at high packet processing rates, previous works \cite{10.1145/3359989.3365406} strongly suggest that performance would not be affected by INDDoS, as for any strategy that can be fully compiled and executed by the Tofino hardware. 

Additionally, we embedded two registers in our P4 program to monitor packet processing time. 
One was placed at the ingress pipeline, and stored the timestamp $t_{in}$ of when a packet entered the switch. 
The timestamp $t_{out}$ was recorded in another register at the egress pipeline when the packet had been processed. 
Per-packet processing time could thus be calculated as $t_{out} - t_{in}$.
This averaged around 200\,ns for simple forwarding, with INDDoS adding an additional 100\,ns or so on top of that, which we deem an acceptable time overhead.

\subsection{Evaluation of resource usage}
Table~\ref{resource} shows the switch resources required by our INDDoS and simple forwarding implementation. 
Considering that each stage can only apply a single atomic action, several stages must be used to carry out INDDoS calculations including simple forwarding, meaning that we require all available stages in the switch to achieve the best detection performance, whereas simple forwarding only requires a small fraction of them (16.67\%). 
Moreover, our strategy uses only 8.33\% of the total available SRAM.
This is because each register cell in BACON Sketch only occupies 1 bit, and the size of Count-min Sketch is small (i.e., only $3\times 1024$) though each of its register cells occupies 32 bits. 
INDDoS needs 56.25\% of total ALUs  for processing the packets. 
This is high since INDDoS not only needs to apply the match-action table for simple packet forwarding (4.2\% of all ALUs), but also maintains two sketches (BACON and auxiliary), which require more arithmetic operations. 
The packet header vector (PHV) size indicates the amounts of packet header information passed across the pipeline stages.
In our case, only 9.90\% of PHV is required for storing this information, compared to 7.30\% for simple forwarding, meaning that INDDoS does not embed much additional temporary metadata in the packet.

It is important to clarify that, even though INDDoS requires all available pipeline stages, it does not consume all the resources in each stage (as shown in the columns SRAM, No. ALUs and PHV size). 
This means that it would be possible to consolidate other tasks (e.g. DDoS mitigation as introduced in Section \ref{sec:ddos_mitigation}) in the same pipeline if properly developed.
Moreover, the relative cost of our solution is expected to be lower if deployed in next-generation programmable hardware (e.g. Tofino 2), which will provide more resources and stages.

\begin{table}
 \caption{Switch resource usage of INDDoS}
\label{resource}
 \centering
\resizebox{0.6\linewidth}{!}{
\begin{tabular}{|c|c|c|c|c|c|}
\hline
\textbf{Strategy}& \textbf{No. stages} & \textbf{SRAM}    & \textbf{No. ALUs} & \textbf{PHV size}\\
\hline
Simple forwarding & 16.67\%  & 2.5\%  & 4.2\% &  7.30\%\\
\hline
\begin{tabular}[x]{@{\vspace{0cm}}c@{}}INDDoS +\\ Simple forwarding\end{tabular} & 100\%  & 8.33\%   & 56.25\% &  9.90\%\\
\hline
\end{tabular}
}
\end{table}

\section{Related works} \label{sec:related}
In this section we describe recent related works and solutions on \emph{(i.)} DDoS detection in Software-Defined Networks, \emph{(ii.)} programmable data plane capabilities with ASICs and \emph{(iii.)} in-network monitoring using programmable switches.

\subsection{DDoS detection in the context of SDN}
Many techniques have already been proposed to detect various kinds of DDoS attacks in SDN networks.
A DDoS attack can be identified according to many different metrics, such as looking for a significant decrease of the normalized entropy in distinct destination IP addresses observed in the network \cite{giotis2014combining}\cite{kalkan2018jess}\cite{wang2015entropy}, or a large number of distinct flows (sequences of packets with the same source IPs) contacting a specific destination host (i.e., per-destination flow cardinality) \cite{yu2013software}\cite{liu2016one}\cite{huang2017sketchvisor}. 
Note that entropy-based DDoS detection can only detect DDoS attacks, but flow cardinality-based DDoS detection is also able to identify the DDoS victims, which allows operators to mitigate the impact on targeted nodes as soon as an attack is detected.
However, the state-of-the-art approaches \cite{yu2013software}\cite{liu2016one}\cite{huang2017sketchvisor} in SDN still need the controller to periodically retrieve the information from the switches for further processing. With our flow-cardinality-based INDDoS approach we make a step further: we exploit data-plane programmable switches to just forward the information on DDoS victims to the controller, offloading the DDoS detection task to the switch.

\subsection{Data plane programmable switches exploiting ASICs}
Most of the existing SDN switches come with very limited (or no) programmability with respect to the data plane functions that can be executed. 
To enable new kinds of functionalities (e.g. support of additional protocols) it is necessary to upgrade the hardware, which requires significant capital expenditure.
Recently, programmable ASICs have been introduced: they ensure standard data plane features (i.e., high-speed switching and forwarding) while offering the possibility of customizing functionalities, if properly programmed through domain-specific programming languages like P4 \cite{bosshart2014p4}.
For instance, programmable switches equipped with Tofino ASIC \cite{Tofino} can always forward packets at line-rate once the P4 program (including innovative features) is compiled and installed in the switches. 
Other programmable chips, like Network Interface Cards (NICs), Field Programmable Gate Arrays (FPGAs) and Network Processing Units (NPUs), cannot currently ensure high throughput and low latency on par with ASICs.
Additionally, in the context of network security, compared to highly-optimized software solutions, such as inline Intrusion Detection Systems (IDSs) \cite{krueger2008architecture}, the throughput ensured by ASICs is orders of magnitude higher and introduces much lower latency ($\sim 50\mu s - 1ms$) \cite{zhang2020poseidon}. This makes programmable ASICs well suited for the implementation of some network monitoring/security tasks, such as the DDoS detection strategy proposed in this paper.

\subsection{In-network monitoring tasks using programmable switches} 
Network monitoring has been widely studied, including in the context dealing with the capability of programmable switches.
Recently, researchers have started embedding network monitoring tasks directly into programmable switches, such as heavy hitter detection \cite{sivaraman2017heavy}, network traffic entropy estimation \cite{ding2019estimating} and entropy-based DDoS detection \cite{lapolli2019offloading}. 
Most of these solutions are based on sketches, probabilistic data structures to track summarized information pertaining large numbers of packets using fixed size memory. 
It has been proven that sketch-based monitoring solutions have a better accuracy/memory trade-off than sampling-based solution, at least over short time scales \cite{yang2018elastic}.
A common feature of these approaches is that the monitoring outcomes gathered from sketches are reported to the controller only when an anomalous event is detected, therefore overcoming the limitations of large communication overhead and latency caused by the interaction between data plane and control plane. 
Unfortunately, another common theme among these works is that, unlike ours, their P4 code was only tested in the (largely resource-unconstrained) P4 Behavioral model \cite{p4simulator} simulator.
Exceptions to this include Tang \emph{et al.}~\cite{tangspreadsketch}, who proposed SpreadSketch, a new sketch for DDoS detection implementable in Tofino-based switches. However, their solution relies on the controller querying the sketch for identifying the attacks, while in our proposal queries occur inside the switch, thus reducing switch-controller interaction needs. 
Dimolianis \emph{et al.} \cite{dimolianis2020multi} presented another in-network DDoS detection scheme working in actual Netronome SmartNICs \cite{netronome}.
Their approach measures three different features: total number of incoming traffic flows, subnet significance and packet symmetry. However, it is only able to identify the subnet under attack, which may limit the accuracy of deployed mitigation measures.

To the best of our knowledge, INDDoS is the first attempt to perform DDoS detection with host victim identification entirely in commodity switches equipped with programmable ASICs, while dealing with and overcoming all the constraints set by the hardware.

\section{Conclusions}\label{sec:conclusion}
In this paper, we proposed a novel in-network DDoS victim identification strategy, INDDoS, based on a new probabilistic data structure we named BACON Sketch, which combines a Direct Bitmap and a Count-min Sketch to estimate the number of distinct flows contacting the same destination.
INDDoS uses a threshold-based rule to identify victims directly in the programmable data plane of switches.
We proved some parametric bounds on the quality of estimations produced by BACON and INDDoS, and implemented them using the P4 language and toolchain in a Edgecore commodity switch with Tofino ASIC.
The analysis of the performance of our solution proves that it can precisely and accurately identify DDoS victims without adversely affecting the packet processing capabilities of the switch.
Moreover, this approach only reports to the controller when a new victim is detected, greatly reducing the communication strain on the monitoring infrastructure.
This work contributes to the ongoing DDoS attack detection and mitigation activities carried on in the GN4-3 project for upgrading the pan-European GÉANT network.


\section*{Acknowledgement}
The research leading to these results has received funding from the GN4-3 project, within the European H2020 R\&I program, Grant Agreement No. 856726. We also want to thank Barefoot Networks for their valuable support.

\bibliographystyle{IEEEtran}
\bibliography{ref}
\vfill

\end{document}